\documentclass[sn-mathphys]{sn-jnl}

\jyear{2023}%

\theoremstyle{thmstyleone}%
\newtheorem{theorem}{Theorem}
\theoremstyle{thmstyletwo}%
\newtheorem{example}{Example}%

\theoremstyle{thmstylethree}%
\newtheorem{definition}{Definition}%

\newtheorem{lemma}[theorem]{Lemma}

\newtheorem{corollary}[theorem]{Corollary}

\newcommand{\tr}{{\mathrm{Tr}}}

\newcommand{\gf}{{\mathbb{F}}}

\newcommand{\wt}{{\mathtt{wt}}}

\newcommand{\cC}{{\mathcal{C}}}

\newcommand{\bc}{{\mathbf{c}}}
\newcommand{\bg}{{\mathbf{g}}}

\usepackage{enumerate}
\usepackage{booktabs}
\usepackage{arydshln}
\UseRawInputEncoding 
\allowdisplaybreaks[4]
\begin{document}

\title[ ]{The augmented codes of a family of linear codes with locality 2}

\author*[1,2]{\fnm{Ziling} \sur{Heng}}\email{zilingheng@chd.edu.cn}

\author[1]{\fnm{Keqing} \sur{Cao}}\email{ckq0009282022@163.com}

\affil*[1]{\orgdiv{School of Science}, \orgname{Chang'an University}, \orgaddress{ 
\city{Xi'an}, \postcode{710064},
\country{China}}}

\affil[2]{\orgdiv{National Mobile Communications Research Laboratory}, \orgname{Southeast University}, \orgaddress{ 
\city{Nanjing}, \postcode{211111},
\country{China}}}


\abstract{Abstract¡ª In this paper, we first generalize the class of linear codes by
Ding and Ding (IEEE TIT, 61(11), pp. 5835-5842, 2015).
Then we mainly study the augmented codes of this generalized class of linear codes.
For one thing, we use Gaussian sums to determine the parameters and weight distributions
of the augmented codes in some cases.
It is shown that the augmented codes are self-orthogonal and have only a few nonzero weights.
For another thing, the locality of the augmented codes is proved to be 2, which indicates
the augmented codes are useful in distributed storage.
Besides, the augmented codes are projective as the minimum distance of
their duals is proved to be 3.  In particular, we obtain several (almost) optimal linear codes and locally recoverable codes.

}




\maketitle
\section{Introduction}\label{sec1}
Let $\gf_q$ be the finite field with $q$ elements, where $q$ is a  power of a prime $p$.
Denote by $\gf_q^*=\gf_q\backslash \{0\}$. Let $\wt(\bc)$ be the Hamming weight of a vector $\bc\in \gf_q^n$.
\subsection{Linear codes}
Linear codes have been widely studied as they have nice application in communication systems and data storage systems.
For a positive integer $n$, if $\cC\subseteq \gf_q^n$ is a $k$-dimensional $\gf_q$-linear subspace of $\gf_q^n$, then it is called an $[n,k,d]$ linear code over $\mathbb{F}_q$, where $d=\min\limits_{\bc \in \cC\backslash \{\textbf{0}\}}\wt(\bc)$ denotes its minimum  distance.
Let $A_{i}:=\sharp \left\{\bc \in \cC: \wt(\bc)=i \right\}$, where $0 \leq i \leq n$.
Then  $(1,A_{1},A_{2}, \cdots ,A_{n})$ is called the weight distribution of $\cC$ and the polynomial
$$
A(z)=1+A_{1}z+A_{2}z^2+ \cdots +A_{n}z^n
$$
is referred to as the weight enumerator of $\cC$.
It is interesting to study the weight distribution of a linear code as  it is used to estimate the error correcting capability and the probability of error detection and correction of the code with respect to some algorithms \cite{K}. Usually, a linear code $\cC$ is said to be a $t$-weight code where $t:=\sharp \{A_i:A_i \neq 0, 0 \leq i \leq n\}$. Linear codes with only a few weights, e.g. two-weight and three weight codes, have been widely studied and constructed in the literature \cite{D1, D2, D3, D4, D5, LC, LY, HZ2, HZ3, HZ4}.

Let $\langle \textbf{a} , \textbf{b} \rangle$ denote the standard inner product of two vectors $\textbf{a},\textbf{b}$ of the same length.
Then we define the dual of an $[n,k]$ linear code $\cC$ over $\gf_q$ as
$$
\mathcal{C}^{\perp}=\left\{ \mathbf{u}^{\perp} \in \mathbb{F}_{q}^{n}: \langle  \mathbf{u}^{\perp}, \mathbf{u} \rangle=0, \mbox{ $\forall$ }\mathbf{u} \in \mathcal{C} \right\}.
$$
It is clear that $\cC^{\perp}$ is an $[n,n-k]$ linear code. In particular, if $\cC \subseteq \cC^{\perp}$,  then $\cC$ is said to be self-orthogonal.
If $\cC=\cC^\perp$, then $\cC$ is said to be self-dual.  Self-orthogonal codes are of particular interest as they have nice applications in quantum codes, lattices and many other fields \cite{CSS, Wan}. The reader is referred to \cite{LH, WL, WY} for some known self-orthogonal codes.

For a linear code $\cC$ over $\gf_q$, if all the codewords have weights divisible by an integer $\Delta$, then $\cC$ is said to be $\Delta$-divisible.
Divisible code was first introduced by Ward in 1981 \cite{Ward}.
 If $\gcd(\Delta,q)= 1$, then $\cC$ is equivalent to a code derived by taking a linear
code over $\gf_q$, repeating each coordinate $\Delta$ times, and adding on enough $0$ entries to make a code whose length is that of $\cC$ \cite{Ward}.
Thus we are only interested in $\Delta$-divisible codes such that $\Delta$ is a power of the characteristic of $\gf_q$.
Divisible codes have many interesting applictions in Galois geometries, subspace codes, partial spreads,
vector space partitions, and Griesmer codes \cite{KK, Kurz, Ward}.

Recently, Li and Heng established a sufficient condition for a divisible code to be self-orthogonal.
\begin{lemma}\cite[Theorem 11]{LH}\label{LH}
Let $q=p^e$, where $p$ is an odd prime. Let $\cC$ be an $[n,k,d]$ linear code over $\gf_q$ with $\mathbf{1} \in\cC$, where $\mathbf{1}$ is all-1 vector of length $n$. If the linear code $\cC$ is $p$-divisible, then $\cC$ is self-orthogonal.
\end{lemma}

\subsection{Locally recoverable codes}
In distributed and cloud storage systems, locally recoverable codes (LRCs for short) are used to recover the data.
LRCs were first introduced by Gopalan et al. in 2012 \cite{G}, and have been implemented in
practice by Microsoft and Facebook \cite{HS, S}. An LRC with locality $r$ is a linear code such that  any symbol
in the encoding is a function of $r$ other symbols. In other words, any symbol can be repaired from
at most $r$ other code symbols. It is known that the repair cost of an $[n,k,d]$ LRC with locality $r\ll k$ is lower
in comparison to MDS codes. The conventional definition of LRC is given as follows.

\begin{definition}\cite{HY}
For an $[n,k,d]$ linear code $\cC$ over $\gf_q$ with generator matrix $G=[\textbf{g}_1,\textbf{g}_2,\cdots,\textbf{g}_n]$,
if each column vector $\textbf{g}_i$ is a linear combination of $l$ $(\leq r)$ other columns of $G$, then $r$ is referred to as the locality of $\cC$ and $\cC$
is called an $[n,k,d; r]_q$-LRC.
\end{definition}

Similarly to the classical linear codes, there also exist some tradeoffs among the parameters of LRCs.
The following provides an upper bound on the dimension of an $[n,k,d; r]_q$-LRC.

\begin{lemma}\cite[Cadambe-Mazumdar bound]{CM}\label{lem-CMbound}
Let $\mathbb{Z}^{+}$ denote the set of all positive integers.
For any $[n,k,d; r]_q$-LRC, we have
\begin{eqnarray}\label{eqn-CMbound}
k \leq \mathop{\min}_{t \in \mathbb{Z}^{+}}\left [rt+k_{\text{opt}}^{(q)}(n-t(r+1),d)\right],
\end{eqnarray}
where $k_{\text{opt}}^{(q)}(n,d)$ is the largest possible dimension of a linear code with length $n$, minimum distance $d$ and alphabet size $q$.
\end{lemma}

The upper bound on the minimum distance of an $[n,k,d; r]_q$-LRC is given as follows.

\begin{lemma}\cite[Singleton-like bound]{G}\label{like}
For any $[n,k,d; r]_q$-LRC, its minimum distance satisfies
\begin{eqnarray}
d \leq n-k-\left \lceil \frac{k}{r} \right \rceil+2.
\end{eqnarray}
\end{lemma}

LRCs are said to be distance-optimal ($d$-optimal for short) when they achieve the Singleton-like bound.
LRCs are said to be almost distance-optimal (almost $d$-optimal for short) when they meet the Singleton-like
bound minus one with equality. LRCs are said to be dimension-optimal ($k$-optimal for short) when they achieve the Cadambe-Mazumdar
bound. LRCs are said to be almost dimension-optimal (almost $k$-optimal for short) when they meet the
Cadambe-Mazumdar bound minus one with equality. An LRC achieveing either of
these two bounds is said to be optimal.
Constructing optimal or almost optimal LRCs has been an interesting research topic in recent years.

\subsection{The objectives of this paper}
To state the objectives of this paper, we recall a well-known construction of linear codes.
 Let $D=\{x_1, x_2, \cdots, x_n\} \subseteq \gf_{q^{m}}$, where $q$ is a power of a prime $p$.
 Define the trace function from $\gf_{q^m}$ to $\gf_{q}$ by
 $$\tr_{q^m/q}(x)=x+x^q+x^{q^2}+\cdots+x^{q^{m-1}},\ x \in \gf_{q^m}.$$
 Define a linear code of length $n$ over $\gf_q$ by
\begin{eqnarray}\label{eq-CD}
\cC_D=\left\{\left(\tr_{q^m/q}(bx_1), \tr_{q^m/q}(bx_2), \cdots, \tr_{q^m/q}(bx_n)\right):b\in \gf_{q^m}\right\}.
\end{eqnarray}
Then the set $D$ is called the defining set of $\cC_D$. Note that $\cC_D$ has dimension at most $m$ and the ordering of the elements in $D$ dose not affect the parameters and weight distribution of $\cC_D$. This general construction was introduced by Ding et al. in \cite{D2, D3} to construct linear codes with only a few weights.

In \cite{D4}, Ding and Ding used the defining set $D_1=\left\{x\in \gf_{p^m}^*:\tr_{p^m/p}(x^2)=0\right\}$ to construct a family of $p$-ary linear codes as
$$\cC_{D_1}=\left\{\left(\tr_{p^m/p}(bx)\right)_{x\in D_1}:b\in \gf_{p^m}\right\}.$$
The parameters and weight distributions of $\cC_{D_1}$ were respectively determined for odd $m$ and even $m$ in \cite{D4}.
Now we generalize this family of linear codes as follows. Let $m_1$ and $m_2$ be positive integers such that $m_1\mid m$ and $m_2\mid m$. Choose the defining set
\begin{eqnarray}\label{D}
D=\left\{x \in \gf_{p^m}:\tr_{p^m/p^{m_1}}(x^2)=0\right\}.
\end{eqnarray}
Define a linear code over $\gf_{p^{m_2}}$ by
$$\cC_{D}=\left\{\left(\tr_{p^m/p^{m_2}}(bx)\right)_{x\in D}:b\in \gf_{p^m}\right\}.$$
Note that $\cC_{D}$ is the same as $\cC_{D_1\cup \{0\}}$ if $m_1=m_2=1$. Hence $\cC_{D}$ is a generalization of $\cC_{D_1\cup \{0\}}$.
Now we define the augmented code of $\cC_{D}$ by
\begin{eqnarray}\label{CD}
\overline{\cC_{D}}=\left\{\left(\tr_{p^m/p^{m_2}}(bx)\right)_{x \in D} + c\mathbf{1}:b \in \gf_{p^m}, c \in \gf_{p^{m_2}}\right\},
\end{eqnarray}
where $\textbf{1}$ is the all-1 vector of the same length as that of $\overline{\cC_{D}}$. Note that $\overline{\cC_{D}}$ may have larger dimension that that of $\cC_{D}$ if
 $\textbf{1}\not\in \cC_{D}$. The augmentation technique was used in \cite{DT} to construct linear codes holding $t$-designs. In \cite{DT}, it was pointed out that determining the minimum distance of the augmented code of a linear code is in general difficult as we require the complete weight information of the original code.

 The main objective of this paper is to study the augmented code $\overline{\cC_{D}}$ in the following two cases:
 \begin{enumerate}
 \item[(i)] $m_2\mid m_1 \mid m$;
 \item[(ii)] $m_1\mid m_2 \mid m$.
 \end{enumerate}
The parameters and weight distributions of $\overline{\cC_{D}}$ are respectively determined in these two cases.
By the weight distribution of $\overline{\cC_{D}}$ and Lemma \ref{LH}, it is shown that  $\overline{\cC_{D}}$ is $p$-divisible and self-orthogonal in these cases.
Besides, the locality of $\overline{\cC_{D}}$ is proved to be 2, which indicates
$\overline{\cC_{D}}$ is useful in distributed storage.
We also prove that $\overline{\cC_{D}}$ is projective as the minimum distance of
its dual is 3.  In particular, we obtain several (almost) optimal linear codes and locally recoverable codes.

\section{Preliminaries}\label{sec2}
In this section, we recall some known results on characters and Gaussian sums over finite fields.
\subsection{Characters over finite fields}
Let $q=p^m$ with $p$ a prime and $m$ a positive integer. Let $\mathbb{C}^*$ be the set of all nonzero complex numbers.
Denote by $\zeta_p$ the primitive $p$-th root of complex unity.
 An additive character of $\gf_q$ is a function $\chi$ from $\gf_q$ to $\mathbb{C}^*$ such that $
\chi(x+y)=\chi(x)\chi(y)$ for any pair $(x,y)\in \gf_q \times \gf_q$. For each $a \in \gf_q$, the function
\begin{eqnarray*}\label{eq1}
\chi_a(x)=\zeta_p^{\tr_{q/p}(ax)},\ x\in \gf_q,
\end{eqnarray*}
defines an additive character of $\gf_q$, where $\tr_{q/p}(x)$ is the trace function from $\gf_q$ to $\gf_p$.
Then the set $\widehat{\gf}_q:=\{\chi_a:a\in \gf_q\}$ consists of all the additive characters of $\gf_q$.
In particular, $\chi_0$ is called the trivial additive character and $\chi_1$ is referred to as the canonical additive character of $\gf_q$.
Besides, we have $\chi_a(x)=\chi_1(ax)$ for $a,x\in \gf_q$. The orthogonal relation of additive characters (see \cite{L}) is given by
\begin{eqnarray*}
\sum_{x\in \gf_q}\chi_a(x)=\begin{cases}
q    &\text{if $a=0,$}\\
0     &\text{otherwise.}
\end{cases}
\end{eqnarray*}

Let $\gf_q^*=\langle\alpha\rangle$. A multiplicative character of $\gf_q$ is a function $\psi$ from $\gf_q^*$ to  $\mathbb{C}^*$ such that
$$\psi(xy)=\psi(x)\psi(y),\ x,y\in \gf_q^*.$$
For each $0\leq j\leq q-2$, all the multiplicative characters of $\gf_q$ can be given by
$$\psi_j(\alpha^k)=\zeta_{q-1}^{jk}$$ for $k=0,1,\cdots,q-2$, where $0\leq j \leq q-2$.
That is to say, the set $\widehat{\gf_q^*}:=\{\psi_j:j=0,1,\cdots,q-2\}$ gives all multiplicative characters of $\gf_q$ and is a multiplicative group of order $q-1$.
In particular, for odd $q$ and $j=(q-1)/2$, $\eta:=\psi_{(q-1)/2}$ is called the quadratic character of $\gf_q$. Besides, $\psi_0$ is called the trivial multiplicative character of $\gf_q$.
The orthogonal relation of multiplicative characters (see \cite{L}) is given by
\begin{eqnarray*}
\sum_{x\in\gf_q^*}\psi_j(x)=\begin{cases}
q-1    &\text{if $j=0,$}\\
0     &\text{if $j \neq 0$.}
\end{cases}
\end{eqnarray*}

\subsection{Gaussian sums}
Let $\psi$ and $\chi$ are multiplicative and additive characters of $\gf_q$, respectively.
The Gaussian sum is defined by
$$G(\psi,\chi)=\sum_{x\in \gf_{q}^*}\psi(x)\chi(x).$$
In particular, $G(\eta,\chi)$ is called the quadratic Gaussian sum.
The explicit values of Gaussian sums are known only for a few special cases. The value of quadratic Gaussian
sum is given as follows.

\begin{lemma}[\cite{L}, Theorem 5.15]\label{lem-2N}
Let $q=p^m$ with $p$ an odd prime. Then
\begin{eqnarray*}
G(\eta,\chi_1)&=&(-1)^{m-1}(\sqrt{-1})^{(\frac{p-1}{2})^2m}\sqrt{q}.
\end{eqnarray*}
\end{lemma}

The value of the Weil sum associated with quadratic polynomial is given by quadratic Gaussian sum as follows.

\begin{lemma}[\cite{L}, Theorem 5.33]\label{lem-weil}
Let $\chi$ be a nontrivial additive character of $\gf_q$ with $q$ an odd prime, and let $f(x)=a_2x^2+a_1x+a_0 \in \gf_q[x]$, where $a_2 \neq 0$. Then
\begin{eqnarray*}
\sum_{c \in \gf_q}\chi\left(f(c)\right) = \chi\left(a_0-a_1^2(4a_2)^{-1}\right)\eta(a_2)G(\eta, \chi).
\end{eqnarray*}
\end{lemma}

The following lemma will be used later.
\begin{lemma}\label{lem-eta}\cite{D4}
Let $\eta$ and $\eta'$ be the quadratic characters of $\gf_q$ and $\gf_p$, respectively.
If $m \geq2$ is even, then $\eta(y)=1$ for each $y\in \gf_p^*$. If $m$ is odd, then $\eta(y)=\eta'(y)$ for each $y\in\gf_p^*$.
\end{lemma}

\section{Some lemmas}\label{sec3}
From now on, let $p$ be an odd prime.
Let $m,m_1,m_2$ be three positive integers such that $m_1\mid m$ and $m_2\mid m$.
Denote by $\chi_1$, $\chi_1'$ and $\chi_1''$ the canonical additive characters of $\gf_{p^m}$, $\gf_{p^{m_1}}$ and $\gf_{p^{m_2}}$, respectively. Let $\eta$, $\eta'$ and $\eta''$ respectively denote the quadratic multiplicative characters of $\gf_{p^m}$, $\gf_{p^{m_1}}$ and $\gf_{p^{m_2}}$.
In this section, we give some lemmas which will be used to prove our main theorems.

\begin{lemma}\label{lem-omega}
Denote by
\begin{eqnarray*}\label{eq3}
\Omega(b,c) = \sum_{z \in \gf_{p^{m_2}}^*}\chi_1''(zc)\sum_{y \in \gf_{p^{m_1}}^{*}}\chi_1(-z^2b^2y)\eta(y),
\end{eqnarray*}
where $b \in \gf_{p^{m}}$ and $c\in\gf_{p^{m_2}}$. The values of $\Omega(b,c)$ are respectively given for the case $m_2\mid m_1$ and the case $m_1\mid m_2$ in the following.
For $m_2\mid m_1$, if $\frac{m}{m_1}$ is odd, then
\begin{eqnarray*}
& &\Omega(b,c)\\
&=& \left\{\begin{array}{ll}
                   0 & \text{if $\tr_{p^m/p^{m_1}}(b^{2})=0$ and $c\in \gf_{p^{m_2}}$,}\\
                   -G(\eta',\chi_1')\eta'(-\tr_{p^m/p^{m_1}}(b^{2})) & \text{if $\tr_{p^m/p^{m_1}}(b^{2})\neq0$ and $c\neq0$,}\\
                   (p^{m_2}-1)G(\eta',\chi_1')\eta'(-\tr_{p^m/p^{m_1}}(b^{2}))& \text{if $\tr_{p^m/p^{m_1}}(b^{2})\neq0$ and $c=0$;}\\
                   \end{array}\right.
\end{eqnarray*}
if $\frac{m}{m_1}$ is even, then
\begin{eqnarray*}
\Omega(b,c)= \begin{cases}
             (p^{m_1}-1)(p^{m_2}-1) & \mbox{if $\tr_{p^m/p^{m_1}}(b^{2})=0$ and $c=0$,}\\
                   -(p^{m_1}-1) & \mbox{if $\tr_{p^m/p^{m_1}}(b^{2})=0$ and $c\neq0$,}\\
                   -(p^{m_2}-1) & \mbox{if $\tr_{p^m/p^{m_1}}(b^{2})\neq0$ and $c=0$,}\\
                   1 & \mbox{if $\tr_{p^m/p^{m_1}}(b^{2})\neq 0$ and $c\neq 0$.}
                  \end{cases}
\end{eqnarray*}
For $m_1\mid m_2$, if $\frac{m}{m_1}$ is odd and $\frac{m_2}{m_1}$ is odd, then
\begin{eqnarray*}
\Omega(b,c)= \left\{\begin{array}{ll}
              0  & \substack{\mbox{if } \tr_{p^m/p^{m_2}}(b^{2})=0,\\ \mbox{and }c\in \gf_{p^{m_2}}}\\
              (p^{m_1}-1)\eta''(-\tr_{p^m/p^{m_2}}(b^{2}))G(\eta'',\chi_1'') & \substack{\mbox{if } \tr_{p^m/p^{m_2}}(b^{2})\neq0, c=0 \mbox{ or }\\ \tr_{p^m/p^{m_2}}(b^{2})\neq0, c\neq0, \Delta=0,}\\
              -G(\eta'',\chi_1'')\eta''(-\tr_{p^m/p^{m_2}}(b^{2})) & \substack{\mbox{if }  \tr_{p^m/p^{m_2}}(b^{2})\neq0, c\neq0,\\ \mbox{and }\Delta \neq0;}\\
                  \end{array}\right.
\end{eqnarray*}
                  if $\frac{m}{m_1}$ is even and $\frac{m_2}{m_1}$ is odd, then
 \begin{eqnarray*}
 \Omega(b,c) = \left\{\begin{array}{ll}
                   (p^{m_1}-1)(p^{m_2}-1) & \text{if $\tr_{p^m/p^{m_2}}(b^{2})=0$, $c=0$,}\\
                   -(p^{m_1}-1) & \substack{\mbox{if } \tr_{p^m/p^{m_2}}(b^{2})=0, c\neq0 \mbox{ or }\\\tr_{p^m/p^{m_2}}(b^{2})\neq0 , c=0 \mbox{ or }\\\tr_{p^m/p^{m_2}}(b^{2})\neq0 , c \neq0,\vartriangle=0,}\\
                   \substack{\eta'\left(\Delta\right)\eta''\left(-\tr_{p^m/p^{m_2}}(b^{2})\right)\\ \times G(\eta'',\chi_1'')G(\eta',\chi_1')
                    -(p^{m_1}-1)} & \mbox{if } \tr_{p^m/p^{m_2}}(b^{2})\neq0,c\neq0,\vartriangle\neq0;\\
                   \end{array}\right.
 \end{eqnarray*}
                  if $\frac{m}{m_1}$ is even and $\frac{m_2}{m_1}$ is even, then
 \begin{eqnarray*}
 \Omega(b,c) = \left\{\begin{array}{ll}
                   (p^{m_1}-1)(p^{m_2}-1) & \text{if $\tr_{p^m/p^{m_2}}(b^{2})=0$, $c=0$,}\\
                   -(p^{m_1}-1) & \mbox{if }\tr_{p^m/p^{m_2}}(b^{2})=0, c\neq0,\\
                   \substack{(p^{m_1}-1)\left( -1+\eta''(-\tr_{p^m/p^{m_2}}(b^{2}))G(\eta'',\chi_1'')\right)} & \substack{\mbox{if } \tr_{p^m/p^{m_2}}(b^{2})\neq0, c=0 \mbox{ or }\\\tr_{p^m/p^{m_2}}(b^{2})\neq0 , c \neq0,\vartriangle=0,}\\
              \substack{-G(\eta'',\chi_1'') \eta''(-\tr_{p^m/p^{m_2}}(b^{2}))\\-(p^{m_1}-1)} & \substack{\mbox{if }  \tr_{p^m/p^{m_2}}(b^{2})\neq0, c\neq0,\\  \mbox{and }\Delta \neq0,}\\
                   \end{array}\right.
 \end{eqnarray*}
 where $\Delta=\tr_{p^{m_2}/p^{m_1}}\left(\frac{c^2}{\tr_{p^m/p^{m_2}}(b^{2})}\right)$.
 \end{lemma}

 \begin{proof}
  Firstly, let $m_2\mid m_1$. Since $\chi_1(x)=\chi_1'(\tr_{p^m/p^{m_1}}(x))$ for $x\in \gf_{p^m}$, we have
  \begin{eqnarray*}
  \Omega(b, c) &=& \sum_{z \in \gf_{p^{m_2}}^*}\chi_1''(zc)\sum_{y \in \gf_{p^{m_1}}^*}\chi_1(-z^{2}b^{2}y)\eta(y) \\
     &=& \sum_{z \in \gf_{p^{m_2}}^*}\chi_1''(zc)\sum_{y \in \gf_{p^{m_1}}^*}\chi_1'(-yz^{2}\tr_{p^m/p^{m_1}}(b^{2}))\eta(y).
 \end{eqnarray*}
  To calculate the value of $\Omega(b, c)$, we consider the following four cases:

  {Case 1.1:} Let $\tr_{p^m/p^{m_1}}(b^{2})=0$ and $c=0$. By Lemma \ref{lem-eta} and the orthogonal relation of multiplicative characters, we have
   \begin{eqnarray*}
    \Omega(b, c) &=& \sum_{z \in \gf_{p^{m_2}}^*}\sum_{y \in \gf_{p^{m_1}}^*}\eta(y)\\
    &=& \begin{cases}
          \sum_{z \in \gf_{p^{m_2}}^*}\sum_{y \in \gf_{p^{m_1}}^*}\eta'(y) & \mbox{if $\frac{m}{m_1}$ is odd} \\
          \sum_{z \in \gf_{p^{m_2}}^*}\sum_{y \in \gf_{p^{m_1}}^*}1 & \mbox{if $\frac{m}{m_1}$ is even}
        \end{cases}\\
    &=& \begin{cases}
          0 & \mbox{if $\frac{m}{m_1}$ is odd,} \\
         (p^{m_1}-1)(p^{m_2}-1) & \mbox{if $\frac{m}{m_2}$ is even}.
        \end{cases}
  \end{eqnarray*}

  {Case 1.2:} Let $\tr_{p^m/p^{m_1}}(b^{2})=0$ and $c\neq0$. By Lemma \ref{lem-eta} and the orthogonal relations of multiplicative and additive characters, we have
   \begin{eqnarray*}
   \Omega(b, c) &=&\sum_{z \in \gf_{p^{m_2}}^*}\chi_1''(zc)\sum_{y \in \gf_{p^{m_1}}^*}\eta(y) \\
    &=& \begin{cases}
           \sum_{z \in \gf_{p^{m_2}}^*}\chi_1''(zc)\sum_{y \in \gf_{p^{m_1}}^*}\eta'(y)& \mbox{if $\frac{m}{m_1}$ is odd} \\
          \sum_{z \in \gf_{p^{m_2}}^*}\chi_1''(zc)\sum_{y \in \gf_{p^{m_1}}^*}1 & \mbox{if $\frac{m}{m_1}$ is even}
        \end{cases}\\
    &=& \begin{cases}
           0 & \mbox{if $\frac{m}{m_1}$ is odd,} \\
           -(p^{m_1}-1) & \mbox{if $\frac{m}{m_1}$ is even}.
        \end{cases}\\
  \end{eqnarray*}

  {Case 1.3:} Let $\tr_{p^m/p^{m_1}}(b^{2}) \neq 0$ and $c=0$. By Lemma \ref{lem-eta} and the orthogonal relation of additive characters, we have
   \begin{eqnarray*}
   \Omega(b, c) &=&\sum_{z \in \gf_{p^{m_2}}^*}\sum_{y \in \gf_{p^{m_1}}^*}\chi_1'(-yz^{2}\tr_{p^m/p^{m_1}}(b^{2}))\eta(y)\\
    &=& \begin{cases}
          \sum_{z \in \gf_{p^{m_2}}^*}\sum_{y \in \gf_{p^{m_1}}^*}\chi_1'(-yz^{2}\tr_{p^m/p^{m_1}}(b^{2}))\eta'(y) & \mbox{if $\frac{m}{m_1}$ is odd} \\
           \sum_{z \in \gf_{p^{m_2}}^*}\sum_{y \in \gf_{p^{m_1}}^*}\chi_1'(-yz^{2}\tr_{p^m/p^{m_1}}(b^{2}))  & \mbox{if $\frac{m}{m_1}$ is even}
        \end{cases}\\
    &=&\left\{\begin{array}{ll}
         \substack{\sum_{z \in \gf_{p^{m_2}}^*}\sum_{y \in
           \gf_{p^{m_1}}^*}\chi_1'(-yz^{2}\tr_{p^m/p^{m_1}}(b^{2}))\\ \times \eta'(-yz^{2}\tr_{p^m/p^{m_1}}(b^{2}))\eta'(-z^{2}\tr_{p^m/p^{m_1}}(b^{2}))}
           & \text{if $\frac{m}{m_1}$ is odd} \\
           -\sum_{z \in \gf_{p^{m_2}}^*}1 & \text{if $\frac{m}{m_1}$ is even}
            \end{array}\right.\\
    &=&\left\{\begin{array}{ll}
         (p^{m_2}-1)G(\eta',\chi_1')\eta'(-\tr_{p^m/p^{m_1}}(b^{2}))
           & \text{if $\frac{m}{m_1}$ is odd,} \\
         -(p^{m_2}-1) & \text{if $\frac{m}{m_1}$ is even.}
          \end{array}\right.\\
   \end{eqnarray*}

   {Case 1.4:} Let $\tr_{p^m/p^{m_1}}(b^{2}) \neq 0$ and $c\neq0$.  By Lemma \ref{lem-eta} and the orthogonal relation of additive characters, we have
  \begin{eqnarray*}
  \Omega(b, c)&=& \left\{\begin{array}{ll}
         \sum_{z \in \gf_{p^{m_2}}^*}\chi_1''(zc)\sum_{y \in \gf_{p^{m_1}}^*}\chi_1'(-yz^{2}\tr_{p^m/p^{m_1}}(b^{2}))\eta'(y)
          & \text{if $\frac{m}{m_1}$ is odd} \\
         \sum_{z \in \gf_{p^{m_2}}^*}\chi_1''(zc)\sum_{y \in \gf_{p^{m_1}}^*}\chi_1'(-yz^{2}\tr_{p^m/p^{m_1}}(b^{2}))
           & \text{if $\frac{m}{m_1}$ is even} \\
          \end{array}\right.\\
    &=& \left\{\begin{array}{ll}
        \substack{\sum_{z \in \gf_{p^{m_2}}^*}\chi_1''(zc)\sum_{y \in
           \gf_{p^{m_1}}^*}\chi_1'(-yz^{2}\tr_{p^m/p^{m_1}}(b^{2}))\\ \times \eta'(-yz^{2}\tr_{p^m/p^{m_1}}(b^{2}))\eta'(-\tr_{p^m/p^{m_1}}(b^{2}))}
           & \text{if $\frac{m}{m_1}$ is odd} \\
         -\sum_{z \in \gf_{p^{m_2}}^*}\chi_1''(zc)
           & \text{if $\frac{m}{m_1}$ is even} \\
           \end{array}\right.\\
    &=&\left\{\begin{array}{ll}
         -G(\eta',\chi_1')\eta'(-\tr_{p^m/p^{m_1}}(b^{2}))
           & \text{if $\frac{m}{m_1}$ is odd,} \\
         1
           & \text{if $\frac{m}{m_1}$ is even.}
          \end{array}\right.\\
   \end{eqnarray*}

 Secondly, let $m_1\mid m_2$. Since $\chi_1(x)=\chi_1''(\tr_{p^m/p^{m_2}}(x))$ for $x\in \gf_{p^m}$, we have
 \begin{eqnarray*}
  \Omega(b, c) &=& \sum_{y \in \gf_{p^{m_1}}^*}\eta(y)\sum_{z \in \gf_{p^{m_2}}^*}\chi_1''(zc)\chi_1(-z^{2}b^{2}y) \\
     &=& \sum_{y \in \gf_{p^{m_1}}^*}\eta(y)\sum_{z \in \gf_{p^{m_2}}^*}\chi_1''(zc)\chi_1''(-z^{2}y\tr_{p^m/p^{m_2}}(b^{2})).
 \end{eqnarray*}
To calculate the value of $\Omega(b, c)$, we consider the following four cases:

  {Case 2.1:} Let $\tr_{p^m/p^{m_2}}(b^{2})=0$ and $c=0$. By Lemma \ref{lem-eta} and the orthogonal relation of multiplicative characters, we have
   \begin{eqnarray*}
   \Omega(b, c) &=& \sum_{y \in \gf_{p^{m_1}}^*}\eta(y)\sum_{z \in \gf_{p^{m_2}}^*}1\\
    &=& \begin{cases}
          \sum_{y \in \gf_{p^{m_1}}^*}\eta'(y)\sum_{z \in \gf_{p^{m_2}}^*}1 & \mbox{if $\frac{m}{m_1}$ is odd} \\
          \sum_{y \in \gf_{p^{m_1}}^*}1\sum_{z \in \gf_{p^{m_2}}^*}1 & \mbox{if $\frac{m}{m_1}$ is even}
        \end{cases}\\
    &=& \begin{cases}
          0 & \mbox{if $\frac{m}{m_1}$ is odd,} \\
         (p^{m_1}-1)(p^{m_2}-1) & \mbox{if $\frac{m}{m_1}$ is even}.
        \end{cases}
  \end{eqnarray*}

  {Case 2.2:} Let $\tr_{p^m/p^{m_2}}(b^{2})=0$ and $c\neq0$. By Lemma \ref{lem-eta} and the orthogonal relations of multiplicative and additive characters, we have
   \begin{eqnarray*}
   \Omega(b, c) &=& \sum_{y \in \gf_{p^{m_1}}^*}\eta(y)\sum_{z \in \gf_{p^{m_2}}^*}\chi_1''(zc)\\
    &=& \begin{cases}
           \sum_{y \in \gf_{p^{m_1}}^*}\eta'(y)\sum_{z \in \gf_{p^{m_2}}^*}\chi_1''(zc)& \mbox{if $\frac{m}{m_1}$ is odd} \\
         \sum_{y \in \gf_{p^{m_1}}^*}\sum_{z \in \gf_{p^{m_2}}^*}\chi_1''(zc) & \mbox{if $\frac{m}{m_1}$ is even}
        \end{cases}\\
    &=& \begin{cases}
           0 & \mbox{if $\frac{m}{m_1}$ is odd,} \\
           -(p^{m_1}-1) & \mbox{if $\frac{m}{m_1}$ is even}.
        \end{cases}\\
  \end{eqnarray*}

  {Case 2.3:} Let $\tr_{p^m/p^{m_2}}(b^{2}) \neq 0$ and $c=0$.  By Lemmas \ref{lem-eta} and \ref{lem-weil}, we have
   \begin{eqnarray*}
   & &\Omega(b, c) \\
   &=&\sum_{y \in \gf_{p^{m_1}}^*}\eta(y)\sum_{z \in \gf_{p^{m_2}}^*}\chi_1''(-z^{2}y\tr_{p^m/p^{m_2}}(b^{2}))\\
    &=& \begin{cases}
           \sum_{y \in \gf_{p^{m_1}}^*}\eta'(y)\sum_{z \in \gf_{p^{m_2}}^*}\chi_1''(-z^{2}y\tr_{p^m/p^{m_2}}(b^{2}))& \mbox{if $\frac{m}{m_1}$ is odd} \\
           \sum_{y \in \gf_{p^{m_1}}^*}\sum_{z \in \gf_{p^{m_2}}^*}\chi_1''(-z^{2}y\tr_{p^m/p^{m_2}}(b^{2}))& \mbox{if $\frac{m}{m_1}$ is even}
        \end{cases}\\
     &=& \begin{cases}
           \sum\limits_{y \in \gf_{p^{m_1}}^*}\eta'(y)\sum\limits_{z \in \gf_{p^{m_2}}}\chi_1''(-z^{2}y\tr_{p^m/p^{m_2}}(b^{2}))& \mbox{if $\frac{m}{m_1}$ is odd} \\
          -(p^{m_1}-1)+\sum\limits_{y \in \gf_{p^{m_1}}^*}\sum\limits_{z \in \gf_{p^{m_2}}}\chi_1''(-z^{2}y\tr_{p^m/p^{m_2}}(b^{2})) & \mbox{if $\frac{m}{m_1}$ is even}
        \end{cases}\\
      &=& \begin{cases}
           \sum\limits_{y \in \gf_{p^{m_1}}^*}\eta'(y)\eta''(-y\tr_{p^m/p^{m_2}}(b^{2}))G(\eta'',\chi_1'')& \mbox{if $\frac{m}{m_1}$ is odd} \\
          -(p^{m_1}-1)+\sum\limits_{y \in \gf_{p^{m_1}}^*}\eta''(-y\tr_{p^m/p^{m_2}}(b^{2}))G(\eta'',\chi_1'') & \mbox{if $\frac{m}{m_1}$ is even}
        \end{cases}\\
        &=& \begin{cases}
           \sum\limits_{y \in \gf_{p^{m_1}}^*}\eta''(-\tr_{p^m/p^{m_2}}(b^{2}))G(\eta'',\chi_1'') & \substack{\mbox{if }\frac{m}{m_1}\mbox{ is odd and}\\ \frac{m_2}{m_1}\mbox{ is odd}}\\
          -(p^{m_1}-1)+\sum\limits_{y \in \gf_{p^{m_1}}^*}\eta'(y)\eta''(-\tr_{p^m/p^{m_2}}(b^{2}))G(\eta'',\chi_1'') & \substack{\mbox{if }\frac{m}{m_1}\mbox{ is even and}\\ \frac{m_2}{m_1}\mbox{ is odd}}\\
                    -(p^{m_1}-1)+\sum\limits_{y \in \gf_{p^{m_1}}^*}\eta''(-\tr_{p^m/p^{m_2}}(b^{2}))G(\eta'',\chi_1'') & \substack{\mbox{if }\frac{m}{m_1}\mbox{ is even and}\\ \frac{m_2}{m_1}\mbox{ is even}}
        \end{cases}\\
    &=& \begin{cases}
           (p^{m_1}-1)\eta''(-\tr_{p^m/p^{m_2}}(b^{2}))G(\eta'',\chi_1'') & \substack{\mbox{if }\frac{m}{m_1}\mbox{ is odd and}\\ \frac{m_2}{m_1}\mbox{ is odd}}\\
          -(p^{m_1}-1) & \substack{\mbox{if }\frac{m}{m_1}\mbox{ is even and}\\ \frac{m_2}{m_1}\mbox{ is odd}}\\
                    -(p^{m_1}-1)+(p^{m_1}-1)\eta''(-\tr_{p^m/p^{m_2}}(b^{2}))G(\eta'',\chi_1'') & \substack{\mbox{if }\frac{m}{m_1}\mbox{ is even and}\\ \frac{m_2}{m_1}\mbox{ is even}}
        \end{cases}\\
         &=& \begin{cases}
           (p^{m_1}-1)\eta''(-\tr_{p^m/p^{m_2}}(b^{2}))G(\eta'',\chi_1'') & \substack{\mbox{if }\frac{m}{m_1}\mbox{ is odd and}\\ \frac{m_2}{m_1}\mbox{ is odd},}\\
          -(p^{m_1}-1) & \substack{\mbox{if }\frac{m}{m_1}\mbox{ is even and}\\ \frac{m_2}{m_1}\mbox{ is odd},}\\
                   (p^{m_1}-1)\left( -1+\eta''(-\tr_{p^m/p^{m_2}}(b^{2}))G(\eta'',\chi_1'')\right) & \substack{\mbox{if }\frac{m}{m_1}\mbox{ is even and}\\ \frac{m_2}{m_1}\mbox{ is even}.}
        \end{cases}\\
   \end{eqnarray*}

   {Case 2.4:} Let $\tr_{p^m/p^{m_2}}(b^{2}) \neq 0$ and $c\neq0$. By Lemma \ref{lem-weil}, we have
   \begin{eqnarray*}
   & &\Omega(b, c) \\
   &=&\sum_{y \in \gf_{p^{m_1}}^*}\eta(y)\sum_{z \in \gf_{p^{m_2}}^*}\chi_1''(zc-z^{2}y\tr_{p^m/p^{m_2}}(b^{2}))\\
   &=&\sum_{y \in \gf_{p^{m_1}}^*}\eta(y)\left(\sum_{z \in \gf_{p^{m_2}}}\chi_1''(zc-z^{2}y\tr_{p^m/p^{m_2}}(b^{2}))-1\right)\\
   &=&\sum_{y \in \gf_{p^{m_1}}^*}\eta(y)\left(\chi_1''\left(\frac{c^2}{4y\tr_{p^m/p^{m_2}}(b^{2})}\right)\eta''(-y\tr_{p^m/p^{m_2}}(b^{2}))G(\eta'',\chi_1'')-1\right)\\
   &=&\left\{\begin{array}{ll}
         G(\eta'',\chi_1'')\sum\limits_{y \in \gf_{p^{m_1}}^*}\eta'(y)\chi_1''\left(\frac{c^2}{4y\tr_{p^m/p^{m_2}}(b^{2})}\right)\eta''(-y\tr_{p^m/p^{m_2}}(b^{2})) & \text{if $\frac{m}{m_1}$ is odd} \\
        \substack{G(\eta'',\chi_1'')\sum\limits_{y \in \gf_{p^{m_1}}^*}\chi_1''\left(\frac{c^2}{4y\tr_{p^m/p^{m_2}}(b^{2})}\right)\eta''(-y\tr_{p^m/p^{m_2}}(b^{2}))\\-(p^{m_1}-1)} & \text{if $\frac{m}{m_1}$ is even} \\
          \end{array}\right.\\
   &=&\left\{\begin{array}{ll}
         G(\eta'',\chi_1'')\sum\limits_{y \in \gf_{p^{m_1}}^*}\chi_1''\left(\frac{yc^2}{\tr_{p^m/p^{m_2}}(b^{2})}\right)\eta''(-\tr_{p^m/p^{m_2}}(b^{2})) & \substack{\mbox{if }\frac{m}{m_1}\mbox{ is odd and}\\ \frac{m_2}{m_1}\mbox{ is odd}}\\
        \substack{G(\eta'',\chi_1'')\sum\limits_{y \in \gf_{p^{m_1}}^*}\chi_1''\left(\frac{yc^2}{\tr_{p^m/p^{m_2}}(b^{2})}\right)\eta'(y)\eta''(-\tr_{p^m/p^{m_2}}(b^{2}))\\-(p^{m_1}-1)} & \substack{\mbox{if }\frac{m}{m_1}\mbox{ is even and }\\\frac{m_2}{m_1}\mbox{ is odd}} \\
        \substack{G(\eta'',\chi_1'')\sum\limits_{y \in \gf_{p^{m_1}}^*}\chi_1''\left(\frac{yc^2}{\tr_{p^m/p^{m_2}}(b^{2})}\right)\eta''(-\tr_{p^m/p^{m_2}}(b^{2}))\\-(p^{m_1}-1)} & \substack{\mbox{if }\frac{m}{m_1}\mbox{ is even and}\\\frac{m_2}{m_1}\mbox{ is even}} \\
          \end{array}\right.\\
   &=&\left\{\begin{array}{ll}
         \substack{G(\eta'',\chi_1'')\sum\limits_{y \in \gf_{p^{m_1}}^*}\chi_1'\left(y\tr_{p^{m_2}/p^{m_1}}\left(\frac{c^2}{\tr_{p^m/p^{m_2}}(b^{2})}\right)\right)\\ \times \eta''(-\tr_{p^m/p^{m_2}}(b^{2}))} & \substack{\mbox{if }\frac{m}{m_1}\mbox{ is odd and}\\ \frac{m_2}{m_1}\mbox{ is odd,}}\\
        \substack{G(\eta'',\chi_1'')\sum\limits_{y \in \gf_{p^{m_1}}^*}\chi_1'\left(y\tr_{p^{m_2}/p^{m_1}}\left(\frac{c^2}{\tr_{p^m/p^{m_2}}(b^{2})}\right)\right)\\ \times\eta'(y)\eta''(-\tr_{p^m/p^{m_2}}(b^{2}))-(p^{m_1}-1)} & \substack{\mbox{if }\frac{m}{m_1}\mbox{ is even and }\\\frac{m_2}{m_1}\mbox{ is odd},} \\
        \substack{G(\eta'',\chi_1'')\sum\limits_{y \in \gf_{p^{m_1}}^*}\chi_1'\left(y\tr_{p^{m_2}/p^{m_1}}\left(\frac{c^2}{\tr_{p^m/p^{m_2}}(b^{2})}\right)\right)\\ \times \eta''(-\tr_{p^m/p^{m_2}}(b^{2}))-(p^{m_1}-1)} & \substack{\mbox{if }\frac{m}{m_1}\mbox{ is even and}\\\frac{m_2}{m_1}\mbox{ is even},} \\
          \end{array}\right.\\
   \end{eqnarray*}
   where we used the substitution $y\mapsto \frac{y}{4}$ in the fifth equality. If $\tr_{p^{m_2}/p^{m_1}}\left(\frac{c^2}{\tr_{p^m/p^{m_2}}(b^{2})}\right)=0$, then
   \begin{eqnarray*}
   \Omega(b,c)&=&\left\{\begin{array}{ll}
         (p^{m_1}-1)G(\eta'',\chi_1'') \eta''(-\tr_{p^m/p^{m_2}}(b^{2})) & \substack{\mbox{if }\frac{m}{m_1}\mbox{ is odd and}\\ \frac{m_2}{m_1}\mbox{ is odd},}\\
        -(p^{m_1}-1) & \substack{\mbox{if }\frac{m}{m_1}\mbox{ is even and }\\\frac{m_2}{m_1}\mbox{ is odd},} \\
       (p^{m_1}-1) G(\eta'',\chi_1'') \eta''(-\tr_{p^m/p^{m_2}}(b^{2}))-(p^{m_1}-1) & \substack{\mbox{if }\frac{m}{m_1}\mbox{ is even and}\\\frac{m_2}{m_1}\mbox{ is even}.} \\
          \end{array}\right.
   \end{eqnarray*}
   If $\tr_{p^{m_2}/p^{m_1}}\left(\frac{c^2}{\tr_{p^m/p^{m_2}}(b^{2})}\right)\neq 0$, then
   \begin{eqnarray*}
  \Omega(b,c) &=&\left\{\begin{array}{ll}
         -G(\eta'',\chi_1'') \eta''(-\tr_{p^m/p^{m_2}}(b^{2})) & \substack{\mbox{if }\frac{m}{m_1}\mbox{ is odd and}\\ \frac{m_2}{m_1}\mbox{ is odd},}\\
        \substack{G(\eta'',\chi_1'')G(\eta',\chi_1')\eta'\left(\tr_{p^{m_2}/p^{m_1}}\left(\frac{c^2}{\tr_{p^m/p^{m_2}}(b^{2})}\right)\right)\\
        \times \eta''(-\tr_{p^m/p^{m_2}}(b^{2}))-(p^{m_1}-1)} & \substack{\mbox{if }\frac{m}{m_1}\mbox{ is even and }\\\frac{m_2}{m_1}\mbox{ is odd},} \\
        -G(\eta'',\chi_1'') \eta''(-\tr_{p^m/p^{m_2}}(b^{2}))-(p^{m_1}-1) & \substack{\mbox{if }\frac{m}{m_1}\mbox{ is even and}\\\frac{m_2}{m_1}\mbox{ is even}.} \\
          \end{array}\right.
   \end{eqnarray*}
   Combining all the cases above, the desired conclusions directly follow.
   \end{proof}

\begin{lemma}\label{lem-N}
Let $p$ be an odd prime. Let $m, m_1, m_2$ be three positive integers with $m_1 \mid m$ and $m_2 \mid m$. Denote by $N_i=\lvert \{b\in \gf_{p^m}^*: \tr_{p^m/p^{m_i}}(b^2)=a\} \rvert$, where $a \in \gf_{p^{m_i}}$ for $i=1,2$. If $\frac{m}{m_i}$ is odd,  then
\begin{eqnarray*}
  N_1 = \begin{cases}
        p^{m-m_1}-1 & \mbox{if $a=0$,} \\
        p^{m-m_1}+\frac{G(\eta,\chi_1)G(\eta',\chi_1')\eta'(-a)}{p^{m_1}}  & \mbox{if $a \in \gf_{p^{m_1}}^*$},
      \end{cases}
\end{eqnarray*}
and
\begin{eqnarray*}
  N_2 = \begin{cases}
        p^{m-m_2}-1 & \mbox{if $a=0$,} \\
        p^{m-m_2}+\frac{G(\eta,\chi_1)G(\eta'',\chi_1'')\eta''(-a)}{p^{m_2}}  & \mbox{if $a \in \gf_{p^{m_2}}^*$}.
      \end{cases}
\end{eqnarray*}
If $\frac{m}{m_i}$ is even for $i=1,2$, then
\begin{eqnarray*}
  N_i = \begin{cases}
          p^{m-m_i}+\frac{(p^{m_i}-1)G(\eta,\chi_1)}{p^{m_i}}-1 & \mbox{if $a=0$,} \\
          p^{m-m_i}-\frac{G(\eta,\chi_1)}{p^{m_i}} & \mbox{if $a \in \gf_{p^{m_i}}^*$.}
        \end{cases}
\end{eqnarray*}
\end{lemma}

\begin{proof}
 We only prove the case that $i=1$. For the case that $i=2$, we can similarly give the proof.
  By the orthogonal relation of additive characters and Lemma \ref{lem-weil}, we have
    \begin{eqnarray*}
  N_1 &=& \frac{1}{p^{m_1}}\sum_{b \in \gf_{p^{m}}^*}\sum_{y \in \gf_{p^{m_1}}}\zeta_p^{\tr_{p^{m_1}/p}(y(\tr_{p^m/p^{m_1}}(b^{2})-a))}\\
    &=& \frac{p^m-1}{p^{m_1}}+\frac{1}{p^{m_1}}\sum_{b \in \gf_{p^{m}}^*}\sum_{y \in \gf_{p^{m_1}}^*}\chi_1(yb^{2})\chi_1'(-ya)\\
    &=&\frac{p^m-1}{p^{m_1}}-\frac{1}{p^{m_1}}\sum_{y \in \gf_{p^{m_1}}^*}\chi_1'(-ya)+\frac{1}{p^{m_1}}\sum_{y \in \gf_{p^{m_1}}^*}\chi_1'(-ya)\sum_{b \in \gf_{p^m}}\chi_1(yb^{2})\\
    &=& \begin{cases}
         \frac{p^{m}-p^{m_1}}{p^{m_1}}+\frac{G(\eta,\chi_1)}{p^{m_1}}\sum_{y \in \gf_{p^{m_1}}^*}\eta(y) & \mbox{if $a=0$} \\
          p^{m-m_1}+\frac{G(\eta,\chi_1)}{p^{m_1}}\sum_{y \in \gf_{p^{m_1}}^*}\chi_1'(-ya)\eta(y) & \mbox{if $a \neq 0$}
        \end{cases}\\
    &=& \begin{cases}
          \frac{p^{m}-p^{m_1}}{p^{m_1}}+\frac{G(\eta,\chi_1)}{p^{m_1}}\sum_{y \in \gf_{p^{m_1}}^*}\eta'(y)& \mbox{if $a=0$ and $\frac{m}{m_1}$ is odd} \\
          \frac{p^{m}-p^{m_1}}{p^{m_1}}+\frac{G(\eta,\chi_1)}{p^{m_1}}\sum_{y \in \gf_{p^{m_1}}^*}1 & \mbox{if $a=0$ and $\frac{m}{m_1}$ is even} \\
           p^{m-m_1}+\frac{G(\eta,\chi_1)}{p^{m_1}}\sum_{y \in \gf_{p^{m_1}}^*}\chi_1'(-ya)\eta'(-ya)\eta'(-a) & \mbox{if $a \neq 0$ and $\frac{m}{m_1}$ is odd}\\
           p^{m-m_1}+\frac{G(\eta,\chi_1)}{p^{m_1}}\sum_{y \in \gf_{p^{m_1}}^*}\chi_1'(-ya) & \mbox{if $a \neq 0$ and $\frac{m}{m_1}$ is even}
        \end{cases}\\
    &=&\begin{cases}
          \frac{p^{m}-p^{m_1}}{p^{m_1}}& \mbox{if $a=0$ and $\frac{m}{m_1}$ is odd,} \\
          \frac{p^{m}-p^{m_1}}{p^{m_1}}+\frac{(p^{m_1}-1)G(\eta,\chi_1)}{p^{m_1}} & \mbox{if $a=0$ and $\frac{m}{m_1}$ is even,} \\
           p^{m-m_1}+\frac{G(\eta,\chi_1)G(\eta',\chi_1')\eta'(-a) }{p^{m_1}} & \mbox{if $a \neq 0$ and $\frac{m}{m_1}$ is odd,}\\
           p^{m-m_1}-\frac{G(\eta,\chi_1)}{p^{m_1}} & \mbox{if $a \neq 0$ and $\frac{m}{m_1}$ is even.}
        \end{cases}\\
  \end{eqnarray*}
Then the desired conclusions follow.
\end{proof}\

\begin{lemma}\label{lem-f}
Let $p$ be an odd prime. Let $m, m_1, m_2$ be three positive integers with $m_1 \mid m_2\mid m$. Denote by $M=\lvert  \{c \in \gf_{p^{m_2}}^* : {\tr_{p^{m_2}/p^{m_1}}(\frac{c^2}{a})}=t\} \lvert$, where $a \in \gf_{p^{m_2}}^*$, $t \in \gf_{p^{m_1}}$. When $\frac{m_2}{m_1}$ is odd,
\begin{eqnarray*}
 M = \begin{cases}
        p^{m_2-m_1}-1 & \mbox{if $t=0$,} \\
        p^{m_2-m_1}+\frac{G(\eta'',\chi_1'')G(\eta',\chi_1')\eta''(-at)}{p^{m_1}} & \mbox{if $t \in \gf_{p^{m_1}}^*$}.
      \end{cases}
\end{eqnarray*}
When $\frac{m_2}{m_1}$ is even,
\begin{eqnarray*}
  M = \begin{cases}
          p^{m_2-m_1}+\frac{(p^{m_1}-1)G(\eta'', \chi_1'')\eta''(a)}{p^{m_1}}-1 & \mbox{if $t=0$,} \\
          p^{m_2-m_1}-\frac{G(\eta'', \chi_1'')\eta''(a)}{p^{m_1}} & \mbox{if $t \in \gf_{p^{m_1}}^*$.}
        \end{cases}
\end{eqnarray*}
\end{lemma}

\begin{proof}
Let $M':=\lvert \{c \in \gf_{p^{m_2}} : {\tr_{p^{m_2}/p^{m_1}}(\frac{c^2}{a})}=t\} \lvert$.
 By the orthogonal relation of additive characters and Lemma \ref{lem-weil}, we have
  \begin{eqnarray*}
   M' &=& \frac{1}{p^{m_1}}\sum_{c \in \gf_{p^{m_2}}}\sum_{y \in
         \gf_{p^{m_1}}}\zeta_p^{\tr_{p^{m_1}}/p(y(\tr_{p^{m_2}/p^{m_1}}(\frac{c^2}{a})-t))}\\
     &=& \frac{1}{p^{m_1}}\sum_{c \in \gf_{p^{m_2}}}\sum_{y \in \gf_{p^{m_1}}}\chi_1''\left(\frac{yc^2}{a}\right)\chi_1'(-yt)\\
     &=&  p^{m_2-m_1}+\frac{1}{p^{m_1}}\sum_{y \in \gf_{p^{m_1}}^*}\chi_1'(-yt)\sum_{c \in \gf_{p^{m_2}}}\chi_1''\left(\frac{yc^2}{a}\right)\\
     &=& p^{m_2-m_1}+\frac{G(\eta'',\chi_1'')}{p^{m_1}}\sum_{y \in \gf_{p^{m_1}}^*}\chi_1'(-yt)\eta''\left(\frac{y}{a}\right)\\
    &=& \begin{cases}
          p^{m_2-m_1}+\frac{G(\eta'',\chi_1'')}{p^{m_1}}\sum_{y \in \gf_{p^{m_1}}^*}\eta'(y)\eta''(a)& \substack{\mbox{if }t=0\mbox{ and}\\\frac{m_2}{m_1}\mbox{ is odd}}\\
          p^{m_2-m_1}+\frac{G(\eta'',\chi_1'')}{p^{m_1}}\sum_{y \in \gf_{p^{m_1}}^*}\eta''(a) & \substack{\mbox{if }t=0\mbox{ and}\\\frac{m_2}{m_1}\mbox{ is even}}\\
          p^{m_2-m_1}+\frac{G(\eta'',\chi_1'')}{p^{m_1}}\sum_{y \in \gf_{p^{m_1}}^*}\chi_1'(-yt)\eta'(-yt)\eta''(-at) & \substack{\mbox{if }t \neq 0\mbox{ and}\\ \frac{m_2}{m_1}\mbox{ is odd}}\\
          p^{m_2-m_1}+\frac{G(\eta'',\chi_1'')}{p^{m_1}}\sum_{y \in \gf_{p^{m_1}}^*}\chi_1'(-yt)\eta''(a) & \substack{\mbox{if }t \neq 0\mbox{ and}\\\frac{m_2}{m_1}\mbox{ is even}}
        \end{cases}\\
    &=& \begin{cases}
          p^{m_2-m_1} & \mbox{if $t=0$ and $\frac{m_2}{m_1}$ is odd,} \\
          p^{m_2-m_1}+\frac{(p^{m_1}-1)G(\eta'', \chi_1'')\eta''(a)}{p^{m_1}} & \mbox{if $t=0$ and $\frac{m_2}{m_1}$ is even,} \\
          p^{m_2-m_1}+\frac{G(\eta'', \chi_1'')G(\eta',\chi_1')\eta''(-at)}{p^{m_1}}& \mbox{if $t \neq 0$ and $\frac{m_2}{m_1}$ is odd,}\\
          p^{m_2-m_1}-\frac{G(\eta'', \chi_1'')\eta''(a)}{p^{m_1}}& \mbox{if $t \neq 0$ and $\frac{m_2}{m_1}$ is even.}\\
        \end{cases}
  \end{eqnarray*}
  Obviously, if $t=0$, then $M=M'-1$; if $t\neq0$, then $M=M'$. Then the desired conclusions follow.
\end{proof}
\section{The weight distribution and self-orthogonality of $\overline{\cC_D}$}\label{sec4}
In this section, we determine the parameters and weight distributions of  $\overline{\cC_D}$ in Equation (\ref{CD}) in two cases.
In both cases, $\overline{\cC_D}$ is proved to be self-orthogonal.
\subsection{The case that $m_2 \mid m_1\mid m$}\label{sec4.1}
In this subsection, we study $\overline{\cC_D}$ for the case that $m_2 \mid m_1\mid m$.
\begin{theorem}\label{th-4.1}
  Let $p$ be an odd prime. Let $m, m_1, m_2$ be three positive integers such that $m_2 \mid m_1\mid m$. Denote by $l_1=({\frac{p-1}{2}})^2\frac{m+m_1}{2}+{m_1}+m-2+\frac{p^{m_1}-1}{2}$ and $l_2=({\frac{p-1}{2}})^2\frac{m}{2}+m-1$. Then we have the following results.
  \begin{itemize}
  \item Let $\frac{m}{m_1}$ be odd. If $\frac{m}{m_1}>1$, then $\overline{\cC_D}$ is an $[p^{m-m_1}, \frac{m}{m_2}+1]$ code over $\gf_{p^{m_2}}$ with weight distribution in Table \ref{tab4.1}. Particularly, if $m>m_1+2m_2$, then $\overline{\cC_D}$ is self-orthogonal.
  \item Let $\frac{m}{m_1}$ be even. If $\frac{m}{m_1}>2$, then $\overline{\cC_D}$ is an $[p^{m-m_1}+\frac{(p^{m_1}-1)(-1)^{l_2}\sqrt[]{p^m}}{p^{m_1}}, \frac{m}{m_2}+1]$ code over $\gf_{p^{m_2}}$ with weight distribution in Table \ref{tab4.2} and  $\overline{\cC_D}$ is self-orthogonal.
  \end{itemize}
\end{theorem}

  \begin{table}[h] \tiny
 \begin{center}
\caption{The weight distribution of $\overline{\cC_{D}}$ in Theorem \ref{th-4.1} ($\frac{m}{m_1}$ is odd).}\label{tab4.1}
\begin{tabular}{@{}ll@{}}
\toprule%
Weight & Frequency  \\
\midrule
$0$ & $1$\\
$p^{m-m_1}$ & $p^{m_2}-1$ \\
$p^{m-m_1}-\frac{p^m}{p^{m_1+m_2}}$ & $(p^{m_2})(p^{m-m_1}-1)$\\
$p^{m-m_1}-\frac{p^{m}-(-1)^{l_1}p^{\frac{m+m_1}{2}}}{p^{m_1+m_2}}$ & $(p^{m_2}-1)\left[\frac{p^{m_1}-1}{2}(p^{m-m_1}+\frac{(-1)^{l_1}p^{\frac{m+m_1}{2}}}{p^{m_1}})\right]$\\
$p^{m-m_1}-\frac{p^{m}+(-1)^{l_1}p^{\frac{m+m_1}{2}}}{p^{m_1+m_2}}$ & $(p^{m_2}-1)\left[\frac{p^{m_1}-1}{2}(p^{m-m_1}-\frac{(-1)^{l_1}p^{\frac{m+m_1}{2}}}{p^{m_1}})\right]$\\
$p^{m-m_1}-\frac{p^{m}}{p^{m_1+m_2}}-\frac{(-1)^{l_1}p^{\frac{m+m_1}{2}}(p^{m_2}-1)}{p^{m_1+m_2}}$ & $\frac{p^{m_1}-1}{2}(p^{m-m_1}+\frac{(-1)^{l_1}p^{\frac{m+m_1}{2}}}{p^{m_1}})$ \\
$p^{m-m_1}-\frac{p^{m}}{p^{m_1+m_2}}+\frac{(-1)^{l_1}p^{\frac{m+m_1}{2}}(p^{m_2}-1)}{p^{m_1+m_2}}$ & $\frac{p^{m_1}-1}{2}(p^{m-m_1}-\frac{(-1)^{l_1}p^{\frac{m+m_1}{2}}}{p^{m_1}})$ \\
\bottomrule
\end{tabular}
\end{center}
\end{table}

\begin{table}[h] \tiny
\begin{center}
\caption{The weight distribution of $\overline{\cC_{D}}$ in Theorem \ref{th-4.1} ($\frac{m}{m_1}$ is even).}\label{tab4.2}
\begin{tabular}{@{}ll@{}}
\toprule%
Weight & Frequency  \\
\midrule
$0$ & $1$\\
$p^{m-m_1}+\frac{(p^{m_1}-1)(-1)^{l_2}\sqrt[]{p^m}}{p^{m_1}}$ & $p^{m_2}-1$ \\
$p^{m-m_1}-\frac{p^m}{p^{m_1+m_2}}$ & $p^{m-m_1}+\frac{(p^{m_1}-1)(-1)^{l_2}\sqrt[]{p^m}}{p^{m_1}}-1$\\
$\frac{p^m(p^{m_2}-1)}{p^{m_1+m_2}}+\frac{(p^{m_1}-1)(-1)^{l_2}\sqrt[]{p^m}}{p^{m_1}}$ & $(p^{m_2}-1)(p^{m-m_1}+\frac{(p^{m_1}-1)(-1)^{l_2}\sqrt[]{p^m}}{p^{m_1}}-1)$\\
$p^{m-m_1}-\frac{p^m}{p^{m_1+m_2}}+\frac{(-1)^{l_2}\sqrt[]{p^m}(p^{m_1+m_2}-p^{m_1})}{p^{m_1+m_2}}$ & $p^m-p^{m-m_1}-\frac{(p^{m_1}-1)(-1)^{l_2}\sqrt[]{p^m}}{p^{m_1}}$\\
$p^{m-m_1}-\frac{p^m}{p^{m_1+m_2}}+\frac{(-1)^{l_2}\sqrt[]{p^m}\left[(p^{m_1}-1)p^{m_2}-p^{m_1}\right]}{p^{m_1+m_2}}$ & $(p^{m_2}-1)(p^m-p^{m-m_1}-\frac{(p^{m_1}-1)(-1)^{l_2}\sqrt[]{p^m}}{p^{m_1}})$ \\
\bottomrule
\end{tabular}
\end{center}
\end{table}
\begin{proof}
Denote by $n=\lvert D\rvert$. By Lemma \ref{lem-N}, we have
\begin{eqnarray}\label{eqn-n}
\nonumber n&=&\left \lvert \left\{x \in \gf_{p^m}:\tr_{p^m/p^{m_1}}(x^2)=0\right\} \right\rvert\\
&=&\left\{\begin{array}{ll}
                                 p^{m-m_1} &\text {if $\frac{m}{m_1}$ is odd,} \\
                                 p^{m-m_1}+\frac{(p^{m_1}-1)G(\eta,\chi_1)}{p^{m_1}} &\text {if $\frac{m}{m_1}$ is even.} \\
                               \end{array}\right.
\end{eqnarray}

For any codeword $\bc_{(b,c)}=(\tr_{p^m/p^{m_2}}(bx)+c)_{x \in D} \in \overline{\cC_{D}}$, by the orthogonal relation of additive characters, we have
\begin{eqnarray}\label{eqn-wt}
\nonumber \wt(\bc_{(b,c)})&=&n-\left\lvert \left\{x \in \gf_{p^m}:\tr_{p^m/p^{m_1}}(x^2)=0\mbox{ and } \tr_{p^m/p^{m_2}}(bx)+c=0\right\}\right\rvert\\
\nonumber &=&n-\frac{1}{p^{m_1+m_2}}\sum_{x\in \gf_{p^m}}\left(\sum_{y\in \gf_{p^{m_1}}}\zeta_{p}^{\tr_{p^{m_1}/p}(y\tr_{p^m/p^{m_1}}(x^2))}\right)\\
\nonumber & &\cdot \left(\sum_{z\in \gf_{p^{m_2}}}\zeta_{p}^{\tr_{p^{m_2}/p}(z(\tr_{p^m/p^{m_2}}(bx)+c))}\right)\\
\nonumber &=&n-\frac{1}{p^{m_1+m_2}}\sum_{x\in \gf_{p^m}}\sum_{y\in \gf_{p^{m_1}}}\sum_{z\in \gf_{p^{m_2}}}\chi_1(yx^2+zbx)\chi_1''(zc)\\
\nonumber &=&n-p^{m-m_1-m_2}-\frac{1}{p^{m_1+m_2}}\sum_{x\in \gf_{p^m}}\sum_{z\in \gf_{p^{m_2}}^*}\chi_1(zbx)\chi_1''(zc)\\
\nonumber & &-\frac{1}{p^{m_1+m_2}}\sum_{x\in \gf_{p^m}}\sum_{y\in \gf_{p^{m_1}}^*}\chi_1(yx^2)\\
& &-\frac{1}{p^{m_1+m_2}}\sum_{x\in \gf_{p^m}}\sum_{y\in \gf_{p^{m_1}}^*}\sum_{z\in \gf_{p^{m_2}}^*}\chi_1(yx^2+zbx)\chi_1''(zc).
\end{eqnarray}
By the orthogonal relation of additive characters, we obtain that
\begin{eqnarray}\label{eqn-sum1}
\nonumber \sum_{x\in \gf_{p^m}}\sum_{z\in \gf_{p^{m_2}}^*}\chi_1(zbx)\chi_1''(zc)&=&\sum_{z\in \gf_{p^{m_2}}^*}\chi_1''(zc)\sum_{x\in \gf_{p^m}}\chi_1(zbx)\\
 &=& \left\{\begin{array}{ll}
                                 p^m(p^{m_2}-1) &\text {if $b=0$ and $c=0$,}\\
                                 -p^m &\text{if $b=0$ and $c\neq 0$,}\\
                                 0 &\text{if $b\neq 0$ and $c\in \gf_{p^{m_2}}$.}\\
                               \end{array}\right.
\end{eqnarray}
By Lemma \ref{lem-weil} and the orthogonal relation of multiplicative characters, we have
\begin{eqnarray}\label{eqn-sum2}
\nonumber \sum_{x\in \gf_{p^m}}\sum_{y\in \gf_{p^{m_1}}^*}\chi_1(yx^2)&=&\sum_{y\in \gf_{p^{m_1}}^*}\sum_{x\in \gf_{p^m}}\chi_1(yx^2)\\
\nonumber &=&G(\eta,\chi_1)\sum_{y\in \gf_{p^{m_1}}^*}\eta(y)\\
&=&\left\{\begin{array}{ll}
0 &\text {if $\frac{m}{m_1}$ is odd,}\\
(p^{m_1}-1)G(\eta,\chi_1) &\text{if $\frac{m}{m_1}$ is even,}\\
\end{array}\right.
\end{eqnarray}
and
\begin{eqnarray}\label{eqn-sum3}
\nonumber & &\sum_{x\in \gf_{p^m}}\sum_{y\in \gf_{p^{m_1}}^*}\sum_{z\in \gf_{p^{m_2}}^*}\chi_1(yx^2+zbx)\chi_1''(zc)\\
\nonumber &=&\sum_{y\in \gf_{p^{m_1}}^*}\sum_{z\in \gf_{p^{m_2}}^*}\chi_1''(zc)\sum_{x\in \gf_{p^m}}\chi_1(yx^2+zbx)\\
\nonumber &=& \sum_{y\in \gf_{p^{m_1}}^*}\sum_{z\in \gf_{p^{m_2}}^*}\chi_1''(zc) \chi_1\left(-z^2b^2(4y)^{-1}\right)\eta(y)G(\eta, \chi_1)\\
\nonumber &=& G(\eta, \chi_1)\sum_{z\in \gf_{p^{m_2}}^*}\chi_1''(zc)\sum_{y\in \gf_{p^{m_1}}^*} \chi_1\left(-z^2b^2y\right)\eta(y)\\
&=&G(\eta, \chi_1)\Omega(b,c),
\end{eqnarray}
where $\Omega(b,c)$ was defined in Lemma \ref{lem-omega}.
Combining Equations (\ref{eqn-n}), (\ref{eqn-wt}), (\ref{eqn-sum1}), (\ref{eqn-sum2}) and (\ref{eqn-sum3}), we directly obtain that
  \begin{eqnarray*}
    & &\wt(\bc_{(b,c)}) \\
    &=& \left\{\begin{array}{ll}
                                0 & \substack{\mbox{if }b=c=0,}\\
                                p^{m-m_1} & \substack{\mbox{if }b=0\mbox{ and }c \neq 0,} \\
                                p^{m-m_1}-\frac{p^m}{p^{m_1+m_2}} & \substack{\mbox{if }b\neq0\mbox{ and }\\ \tr_{p^m/p^{m_1}}(b^{2})=0,}  \\
                                p^{m-m_1}-\frac{p^{m}-G(\eta,\chi_1)G(\eta',\chi_1')\eta'(-1)}{p^{m_1+m_2}}& \substack{\mbox{if }b\neq0\mbox{, }c \neq 0 \mbox{ and }\\ \eta'(\tr_{p^m/p^{m_1}}(b^{2}))=1,}  \\
                                p^{m-m_1}-\frac{p^{m}+G(\eta,\chi_1)G(\eta',\chi_1')\eta'(-1)}{p^{m_1+m_2}} & \substack{\mbox{if }b\neq0\mbox{, }c \neq 0\mbox{ and }\\ \eta'(\tr_{p^m/p^{m_1}}(b^{2}))=-1,}  \\
                                p^{m-m_1}-\frac{p^{m}}{p^{m_1+m_2}}-\frac{G(\eta,\chi_1)G(\eta',\chi_1')\eta'(-1)(p^{m_2}-1)}{p^{m_1+m_2}} & \substack{\mbox{if }b\neq0\mbox{, } c=0,\mbox{ and }\\ \eta'(\tr_{p^m/p^{m_1}}(b^{2}))=1,}  \\
                                p^{m-m_1}-\frac{p^{m}}{p^{m_1+m_2}}+\frac{G(\eta,\chi_1)G(\eta',\chi_1')\eta'(-1)(p^{m_2}-1)}{p^{m_1+m_2}} & \substack{\mbox{if }b\neq0\mbox{, } c=0,\mbox{ and }\\  \eta'(\tr_{p^m/p^{m_1}}(b^{2}))=-1,}  \\
     \end{array} \right.
  \end{eqnarray*}
  for odd $\frac{m}{m_1}$, and
    \begin{eqnarray*}
    \wt(\bc_{(b,c)}) &=& \left\{\begin{array}{ll}
                                 0 & \substack{\mbox{if }b=c=0,}  \\
                                 p^{m-m_1}+\frac{(p^{m_1}-1)G(\eta,\chi_1)}{p^{m_1}} & \substack{\mbox{if }b=0\mbox{ and }c \neq 0,}  \\
                                p^{m-m_1}-\frac{p^m}{p^{m_1+m_2}} & \substack{\mbox{if }b\neq0, c=0\mbox{ and}\\ \tr_{p^m/p^{m_1}}(b^{2})=0,}  \\
                                \frac{p^m(p^{m_2}-1)}{p^{m_1+m_2}}+\frac{(p^{m_1}-1)G(\eta,\chi_1)}{p^{m_1}} & \substack{\mbox{if }b\neq0, c \neq 0
                                \mbox{ and}\\ \tr_{p^m/p^{m_1}}(b^{2})=0,}  \\
                                p^{m-m_1}-\frac{p^m}{p^{m_1+m_2}}+\frac{G(\eta,\chi_1)(p^{m_1+m_2}-p^{m_1})}{p^{m_1+m_2}} & \substack{\mbox{if }b\neq0, c=0\mbox{ and}\\ \tr_{p^m/p^{m_1}}(b^{2})\neq0,}  \\
                               p^{m-m_1}-\frac{p^m}{p^{m_1+m_2}}+\frac{G(\eta,\chi_1)\left[(p^{m_1}-1)p^{m_2}-p^{m_1}\right]}{p^{m_1+m_2}} & \substack{\mbox{if }b\neq0, c\neq0 \mbox{ and}\\ \tr_{p^m/p^{m_1}}(b^{2})\neq0,}  \\
                                \end{array} \right.
\end{eqnarray*}
for even $\frac{m}{m_1}$.
In the following, we will determine the frequency of the nonzero Hamming weight of $\overline{\cC_D}$ in two cases:

{Case 1:} Let $\frac{m}{m_1}$ be odd. Denote by $w_1=p^{m-m_1}$, $w_2=p^{m-m_1}-\frac{p^m}{p^{m_1+m_2}}$, $w_3=p^{m-m_1}-\frac{p^{m}-G(\eta,\chi_1)G(\eta',\chi_1')\eta'(-1)}{p^{m_1+m_2}}$,
   $w_4=p^{m-m_1}-\frac{p^{m}+G(\eta,\chi_1)G(\eta',\chi_1')\eta'(-1)}{p^{m_1+m_2}}$,
    $w_5= p^{m-m_1}-\frac{p^{m}}{p^{m_1+m_2}}-\frac{G(\eta,\chi_1)G(\eta',\chi_1')\eta'(-1)(p^{m_2}-1)}{p^{m_1+m_2}}$.
    $w_6= p^{m-m_1}-\frac{p^{m}}{p^{m_1+m_2}}+\frac{G(\eta,\chi_1)G(\eta',\chi_1')\eta'(-1)(p^{m_2}-1)}{p^{m_1+m_2}}$.
Now we determine the frequency $A_{w_i}$, $1 \leq i \leq 6$. It is obvious that $A_{w_1}=p^{m_2}-1$. Let $\gf_{p^m}^*=\langle\alpha\rangle$ and $\gf_{p^{m_1}}^*=\langle\beta\rangle$. By Lemma \ref{lem-N}, we have
\begin{eqnarray*}
  A_{w_2} &=& p^{m_2}\left\lvert \left\{b \in \gf_{p^{m}}^* : \tr_{p^m/p^{m_1}}(b^{2})=0 \right\} \right\rvert=p^{m_2}(p^{m-m_1}-1),\\
  A_{w_5} &=& \left\lvert \left\{b \in \gf_{p^m}^*:\eta'(\tr_{p^m/p^{m_1}}(b^{2}))=1\right\} \right\rvert\\
  &=& \left\lvert \left\{b \in \gf_{p^m}^* : \tr_{p^m/p^{m_1}}(b^{2}) \in \langle\beta^2\rangle\right\} \right\rvert\\
  &=& \left(p^{m-m_1}+\frac{G(\eta,\chi_1)G(\eta',\chi_1')\eta'(-1)}{p^{m_1}}\right) \cdot \frac{p^{m_1}-1}{2},\\
\end{eqnarray*}
and
\begin{eqnarray*}
  A_{w_6} &=& \left\lvert \left\{b \in \gf_{p^m}^*:\eta'(\tr_{p^m/p^{m_1}}(b^{2}))=-1\right\} \right\rvert\\
  &=& \left\lvert \left\{b \in\in \gf_{p^m}^*: \tr_{p^m/p^{m_1}}(b^{2}) \in \beta\langle\beta^2\rangle\right\} \right\rvert\\
  &=& \left(p^{m-m_1}-\frac{G(\eta,\chi_1)G(\eta',\chi_1')\eta'(-1)}{p^{m_1}}\right) \cdot \frac{p^{m_1}-1}{2}.\\
\end{eqnarray*}
It is obvious that
\begin{eqnarray*}
A_{w_3}&=&(p^{m_2}-1)A_{w_5}\\
&=&(p^{m_2}-1)\left[\left(p^{m-m_1}+\frac{G(\eta,\chi_1)G(\eta',\chi_1')\eta'(-1)}{p^{m_1}}\right) \cdot \frac{p^{m_1}-1}{2}\right]
\end{eqnarray*}
 and
 \begin{eqnarray*}
 A_{w_4}&=&(p^{m_2}-1)A_{w_6}\\
 &=&(p^{m_2}-1)\left[\left(p^{m-m_1}-\frac{G(\eta,\chi_1)G(\eta',\chi_1')\eta'(-1)}{p^{m_1}}\right) \cdot \frac{p^{m_1}-1}{2}\right].
 \end{eqnarray*}
Besides, by Lemma \ref{lem-2N}, we have
\begin{eqnarray*}
  G(\eta,\chi_1)G(\eta',\chi_1')\eta'(-1)=(-1)^{l_1}p^{\frac{m+m_1}{2}}
\end{eqnarray*}
where $l_1=({\frac{p-1}{2}})^2\frac{m+m_1}{2}+{m_1}+m-2+\frac{p^{m_1}-1}{2}$ .
Then the weight distribution of $\overline{\cC_D}$ for odd $\frac{m}{m_1}$ directly follows in Table \ref{tab4.1}.
If $m>m_1$, then the dimension of $\overline{\cC_D}$ is $\frac{m}{m_2}+1$ as the zero codeword occurs only once when $(b,c)$ runs through $\gf_{p^m} \times \gf_{p^{m_2}}$.

{Case 2:} Let $\frac{m}{m_1}$ be even.
Denote by $w_1'=p^{m-m_1}+\frac{(p^{m_1}-1)G(\eta,\chi_1)}{p^{m_1}}$, $w_2'=p^{m-m_1}-\frac{p^m}{p^{m_1+m_2}}$, $w_3'=\frac{p^m(p^{m_2}-1)}{p^{m_1+m_2}}+\frac{(p^{m_1}-1)G(\eta,\chi_1)}{p^{m_1}}$,
   $w_4'= p^{m-m_1}-\frac{p^m}{p^{m_1+m_2}}+\frac{G(\eta,\chi_1)(p^{m_1+m_2}-p^{m_1})}{p^{m_1+m_2}}$,
    $w_5'= p^{m-m_1}-\frac{p^m}{p^{m_1+m_2}}+\frac{G(\eta,\chi_1)\left[(p^{m_1}-1)p^{m_2}-p^{m_1}\right]}{p^{m_1+m_2}}$.
Now we determine the frequency $A_{w_i'}$, $1 \leq i \leq 5$. It is obvious that $A_{w_1'}=p^{m_2}-1$. By Lemma \ref{lem-N},we have
\begin{eqnarray*}
  A_{w_2'} = \left\lvert \{b \in \gf_{p^{m}}^* : \tr_{p^m/p^{m_1}}(b^{2})=0\}\right \rvert=p^{m-m_1}+\frac{(p^{m_1}-1)G(\eta,\chi_1)}{p^{m_1}}-1
\end{eqnarray*}
and $A_{w_4'}=p^{m}-1-A_{w_2'}=p^m-p^{m-m_1}+\frac{(p^{m_1}-1)G(\eta,\chi_1)}{p^{m_1}}$. Furthermore, we deduce that $A_{w_3'}=(p^{m_2}-1)A_{w_2'}=(p^{m_2}-1)\left(p^{m-m_1}+\frac{(p^{m_1}-1)G(\eta,\chi_1)}{p^{m_1}}-1\right)$ and $A_{w_5'}=(p^{m_2}-1)A_{w_4'}=(p^{m_2}-1)\left(p^m-p^{m-m_1}+\frac{(p^{m_1}-1)G(\eta,\chi_1)}{p^{m_1}}\right)$.
By Lemma  \ref{lem-2N}, we have
\begin{eqnarray*}
  G(\eta,\chi_1)=(-1)^{l_2}\sqrt{p^m},
\end{eqnarray*}
where $l_2=({\frac{p-1}{2}})^2\frac{m}{2}+m-1$.
Then the weight distribution of $\overline{\cC_D}$ for even $\frac{m}{m_1}$  directly follows in Table \ref{tab4.2}. If $m>2m_1$, then the dimension of $\overline{\cC_D}$ is $\frac{m}{m_2}+1$ as the zero codeword occurs only once when $(b,c)$ runs through $\gf_{p^m} \times \gf_{p^{m_2}}$.

The self-orthogonality of $\overline{\cC_D}$ directly follows from the divisibility of $\overline{\cC_D}$ and Lemma \ref{LH}. Then the desired conclusions follow.
\end{proof}

We use Magma to give two examples in the following which are consistent with the results in Theorem \ref{th-4.1}.

\begin{example}
Let $(p,m,m_1,m_2)=(3,6,2,1)$. Then the code $\overline{\cC_{D}}$ in Theorem \ref{th-4.1} has parameters $[81,7,48]$ and weight enumerator $1+360z^{48}+576z^{51}+240z^{54}+720z^{57}+288z^{60}+2z^{81}$. Besides, $\overline{\cC_{D}}$ is $3$-divisible and self-orthogonal.
\end{example}

\begin{example}
Let $(p,m,m_1,m_2)=(3,8,2,1)$. Then the code $\overline{\cC_{D}}$  in Theorem \ref{th-4.1} has parameters $[657,9,414]$ and weight enumerator $1+1312z^{414}+5904z^{432}+11808z^{441}+656z^{486}+2z^{657}$. Besides, $\overline{\cC_{D}}$ is $3$-divisible and self-orthogonal.
\end{example}

\subsection{The case that $m_1 \mid m_2\mid m$}\label{sec4.1}
In this subsection, we study the code $\overline{\cC_D}$ for the case that $m_1 \mid m_2\mid m$.
\begin{theorem}\label{th-4.2}
  Let $p$ be an odd prime. Let $m, m_1, m_2$ be three positive integers such that $m_1 \mid m_2\mid m$. Denote by $l_3=({\frac{p-1}{2}})^2\frac{m+m_2}{2}+{m_2}+m-2$, $l_4=({\frac{p-1}{2}})^2\frac{m+m_2}{2}+{m_2}+m-2+\frac{p^{m_2}-1}{2}$,
  $l_5=({\frac{p-1}{2}})^2\frac{m+m_1+m_2}{2}+m+{m_1}+{m_2}-3+\frac{p^{m_2}-1}{2}$, $l_6=({\frac{p-1}{2}})^2\frac{m_1+m_2}{2}+{m_1}+{m_2}-2+\frac{p^{m_2}-1}{2}$,
  $l_7=({\frac{p-1}{2}})^2\frac{m_2}{2}+{m_2}-1$. Its parameters and weight distributions are respectively given in four cases as follows.
    \begin{itemize}
  \item Let $\frac{m}{m_1}$ be odd and $\frac{m_2}{m_1}$ be odd. If $m>m_1$ then $\overline{\cC_D}$ is an $[p^{m-m_1}, \frac{m}{m_2}+1]$ linear code over $\gf_{p^{m_2}}$ with weight distribution in Table \ref{tab5.1}.
  Particularly, if $m>2m_1+m_2$, then  $\overline{\cC_D}$ is self-orthogonal.
  \item Let $\frac{m}{m_1}$ be even and $\frac{m_2}{m_1}$ be odd. If $m>2m_1$, then $\overline{\cC_D}$ is an $[p^{m-m_1}+\frac{(p^{m_1}-1)(-1)^{l_2}\sqrt[]{p^m}}{p^{m_1}}, \frac{m}{m_2}+1]$ code over $\gf_{p^{m_2}}$ with weight distribution in Table \ref{tab5.2}.   Particularly, if $m>m_1+m_2$, then  $\overline{\cC_D}$ is self-orthogonal.
  \item Let $\frac{m}{m_1}$ be even, $\frac{m_2}{m_1}$ be even and $\frac{m}{m_2}$ be even.  If $m>2m_1$, then $\overline{\cC_D}$ is an $[p^{m-m_1}+\frac{(p^{m_1}-1)(-1)^{l_2}\sqrt[]{p^m}}{p^{m_1}}, \frac{m}{m_2}+1]$ code over $\gf_{p^{m_2}}$ with weight distribution in Table \ref{tab5.31}. Particularly, if $m>2m_1+m_2$, then  $\overline{\cC_D}$ is self-orthogonal.
  \item Let $\frac{m}{m_1}$ be even, $\frac{m_2}{m_1}$ be even and $\frac{m}{m_2}$ be odd.  If $m>2m_1$, then $\overline{\cC_D}$ is an $[p^{m-m_1}+\frac{(p^{m_1}-1)(-1)^{l_2}\sqrt[]{p^m}}{p^{m_1}}, \frac{m}{m_2}+1]$ code over $\gf_{p^{m_2}}$ with weight distribution in Table \ref{tab5.32}. Particularly, if $m>2m_1+m_2$, then  $\overline{\cC_D}$ is self-orthogonal.
  \end{itemize}
  \end{theorem}

\begin{table}[h]
\tiny
\begin{center}
\caption{The weight distribution of $\overline{\cC_{D}}$ in Theorem \ref{th-4.2} ($\frac{m}{m_1}$ is odd and $\frac{m_2}{m_1}$ is odd).}\label{tab5.1}
\begin{tabular}{@{}ll@{}}
\toprule%
Weight & Frequency  \\
\midrule
$0$ & $1$\\
$p^{m-m_1}$ & $p^{m_2}-1$ \\
$p^{m-m_1}-\frac{p^m}{p^{m_1+m_2}}$ & $p^{m_2}(p^{m-m_2}-1)$\\
$p^{m-m_1}-\frac{p^m+(p^{m_1}-1)(-1)^{l_4}p^{\frac{m+m_2}{2}}}{p^{m_1+m_2}}$ & $\frac{p^{m_2-m_1}\cdot(p^{m_2}-1)}{2}(p^{m-m_2}+\frac{(-1)^{l_4}p^{\frac{m+m_2}{2}}}{p^{m_2}})$\\
$p^{m-m_1}-\frac{p^m-(p^{m_1}-1)(-1)^{l_4}p^{\frac{m+m_2}{2}}}{p^{m_1+m_2}}$ &
$\frac{p^{m_2-m_1}\cdot(p^{m_2}-1)}{2}(p^{m-m_2}-\frac{(-1)^{l_4}p^{\frac{m+m_2}{2}}}{p^{m_2}})$\\
$p^{m-m_1}-\frac{p^m-(-1)^{l_4}p^{\frac{m+m_2}{2}}}{p^{m_1+m_2}}$ & $\frac{(p^{m_2}-p^{m_2-m_1})(p^{m_2}-1)}{2}(p^{m-m_2}+\frac{(-1)^{l_4}p^{\frac{m+m_2}{2}}}{p^{m_2}})$\\
$p^{m-m_1}-\frac{p^m+(-1)^{l_4}p^{\frac{m+m_2}{2}}}{p^{m_1+m_2}}$ &
$\frac{(p^{m_2}-p^{m_2-m_1})(p^{m_2}-1)}{2}(p^{m-m_2}-\frac{(-1)^{l_4}p^{\frac{m+m_2}{2}}}{p^{m_2}})$\\
\bottomrule
\end{tabular}
\end{center}
\end{table}

\begin{table}[h]
\tiny
\begin{center}
\caption{The weight distribution of $\overline{\cC_{D}}$ in Theorem \ref{th-4.2} ($\frac{m}{m_1}$ is even and $\frac{m_2}{m_1}$ is odd).}\label{tab5.2}
\begin{tabular}{@{}ll@{}}
\toprule%
Weight & Frequency  \\
\midrule
$0$ & $1$\\
$p^{m-m_1}+\frac{(p^{m_1}-1)(-1)^{l_2}\sqrt[]{p^m}}{p^{m_1}}$ & $p^{m_2}-1$ \\
$p^{m-m_1}-\frac{p^m}{p^{m_1+m_2}}$ & $p^{m-m_2}+\frac{(p^{m_2}-1)(-1)^{l_2}\sqrt[]{p^m}}{p^{m_2}}-1$\\
$p^{m-m_1}-\frac{p^m}{p^{m_1+m_2}}+\frac{(-1)^{l_2}\sqrt[]{p^m}(p^{m_1}-1)}{p^{m_1}}$ & $(p^{m_2}-1)(p^{m-m_2}+\frac{(p^{m_2}-1)(-1)^{l_2}\sqrt[]{p^m}}{p^{m_2}}-1)$\\
& $+p^{m_2-m_1}(p^m-p^{m-m_2}-\frac{(-1)^{l_2}\sqrt[]{p^m}(p^{m_2}-1)}{p^{m_2}})$\\
$p^{m-m_1}+\frac{(-1)^{l_2}\sqrt[]{p^m}(p^{m_1}-1)}{p^{m_1}}-\frac{p^m+(-1)^{l_5}p^{\frac{m+m_1+m_2}{2}}}{p^{m_1+m_2}}$ & $\substack{\frac{(p^{m_2}-1)(p^{m_1}-1)}{2}(p^{m_2-m_1}+\frac{(-1)^{l_6}p^{\frac{m_1+m_2}{2}}}{p^{m_1}})\\ \cdot(p^{m-m_2}-\frac{(-1)^{l_2}\sqrt[]{p^m}}{p^{m_2}})}$\\
$p^{m-m_1}+\frac{(-1)^{l_2}\sqrt[]{p^m}(p^{m_1}-1)}{p^{m_1}}-\frac{p^m-(-1)^{l_5}p^{\frac{m+m_1+m_2}{2}}}{p^{m_1+m_2}}$ &
$\substack{\frac{(p^{m_2}-1)(p^{m_1}-1)}{2}(p^{m_2-m_1}-\frac{(-1)^{l_6}p^{\frac{m_1+m_2}{2}}}{p^{m_1}})\\ \cdot(p^{m-m_2}-\frac{(-1)^{l_2}\sqrt[]{p^m}}{p^{m_2}})}$\\

\bottomrule
\end{tabular}
\end{center}
\end{table}

\begin{table}[h]
\tiny
\begin{center}
\caption{The weight distribution of $\overline{\cC_{D}}$ in Theorem \ref{th-4.2} ($\frac{m}{m_1}$ is even, $\frac{m_2}{m_1}$ is even and $\frac{m}{m_2}$ is even).}\label{tab5.31}
\begin{tabular}{@{}ll@{}}
\toprule%
Weight & Frequency  \\
\midrule
$0$ & $1$\\
$p^{m-m_1}+\frac{(p^{m_1}-1)(-1)^{l_2}\sqrt[]{p^m}}{p^{m_1}}$ & $p^{m_2}-1$ \\
$p^{m-m_1}-\frac{p^m}{p^{m_1+m_2}}$ & $p^{m-m_2}+\frac{(p^{m_2}-1)(-1)^{l_2}\sqrt[]{p^m}}{p^{m_2}}-1$\\
$p^{m-m_1}+\frac{(-1)^{l_2}\sqrt[]{p^m}(p^{m_1}-1)}{p^{m_1}}-\frac{p^m+(p^{m_1}-1)(-1)^{l_4}p^{\frac{m+m_2}{2}}}{p^{m_1+m_2}}$ &
$\substack{\frac{(p^{m_2}-1)}{2}(p^{m-m_2}-\frac{(-1)^{l_2}\sqrt[]{p^m}}{p^{m_2}}) \\ \cdot(p^{m_2-m_1}+\frac{(-1)^{l_7}\sqrt[]{p^{m_2}}(p^{m_1}-1)}{p^{m_1}})}$\\
$p^{m-m_1}+\frac{(-1)^{l_2}\sqrt[]{p^m}(p^{m_1}-1)}{p^{m_1}}-\frac{p^m-(p^{m_1}-1)(-1)^{l_4}p^{\frac{m+m_2}{2}}}{p^{m_1+m_2}}$ &
$\substack{\frac{(p^{m_2}-1)}{2}(p^{m-m_2}-\frac{(-1)^{l_2}\sqrt[]{p^m}}{p^{m_2}}) \\ \cdot(p^{m_2-m_1}-\frac{(-1)^{l_7}\sqrt[]{p^{m_2}}(p^{m_1}-1)}{p^{m_1}})}$\\
$p^{m-m_1}+\frac{(-1)^{l_2}\sqrt[]{p^m}(p^{m_1}-1)}{p^{m_1}}-\frac{p^m}{p^{m_1+m_2}}$ &
$(p^{m-m_2}+\frac{(p^{m_2}-1)(-1)^{l_2}\sqrt[]{p^m}}{p^{m_2}}-1)(p^{m_2}-1)$\\
$p^{m-m_1}+\frac{(-1)^{l_2}\sqrt[]{p^m}(p^{m_1}-1)}{p^{m_1}}-\frac{p^m-(-1)^{l_4}p^{\frac{m+m_2}{2}}}{p^{m_1+m_2}}$ &
$\substack{\frac{(p^{m_2}-1)}{2}(p^{m-m_2}-\frac{(-1)^{l_2}\sqrt[]{p^m}}{p^{m_2}}) \\ \cdot(p^{m_2}-p^{m_2-m_1}-\frac{(-1)^{l_7}\sqrt[]{p^{m_2}}(p^{m_1}-1)}{p^{m_1}})}$\\
$p^{m-m_1}+\frac{(-1)^{l_2}\sqrt[]{p^m}(p^{m_1}-1)}{p^{m_1}}-\frac{p^m+(-1)^{l_4}p^{\frac{m+m_2}{2}}}{p^{m_1+m_2}}$ &
$\substack{\frac{(p^{m_2}-1)}{2}(p^{m-m_2}-\frac{(-1)^{l_2}\sqrt[]{p^m}}{p^{m_2}}) \\ \cdot(p^{m_2}-p^{m_2-m_1}+\frac{(-1)^{l_7}\sqrt[]{p^{m_2}}(p^{m_1}-1)}{p^{m_1}})}$\\
\bottomrule
\end{tabular}
\end{center}
\end{table}

\begin{table}[h]
\tiny
\begin{center}
\caption{The weight distribution of $\overline{\cC_{D}}$ in Theorem \ref{th-4.2} ($\frac{m}{m_1}$ is even, $\frac{m_2}{m_1}$ is even and $\frac{m}{m_2}$ is odd).}\label{tab5.32}
\begin{tabular}{@{}ll@{}}
\toprule%
Weight & Frequency  \\
\midrule
$0$ & $1$\\
$p^{m-m_1}+\frac{(p^{m_1}-1)(-1)^{l_2}\sqrt[]{p^m}}{p^{m_1}}$ & $p^{m_2}-1$ \\
$p^{m-m_1}-\frac{p^m}{p^{m_1+m_2}}$ & $p^{m-m_2}-1$\\
$p^{m-m_1}+\frac{(-1)^{l_2}\sqrt[]{p^m}(p^{m_1}-1)}{p^{m_1}}-\frac{p^m+(p^{m_1}-1)(-1)^{l_4}p^{\frac{m+m_2}{2}}}{p^{m_1+m_2}}$ &
$\substack{\frac{(p^{m_2}-1)}{2}(p^{m-m_2}+\frac{(-1)^{l_4}p^{\frac{m+m_2}{2}}}{p^{m_2}}) \\ \cdot(p^{m_2-m_1}+\frac{(-1)^{l_7}\sqrt[]{p^{m_2}}(p^{m_1}-1)}{p^{m_1}})}$\\
$p^{m-m_1}+\frac{(-1)^{l_2}\sqrt[]{p^m}(p^{m_1}-1)}{p^{m_1}}-\frac{p^m-(p^{m_1}-1)(-1)^{l_4}p^{\frac{m+m_2}{2}}}{p^{m_1+m_2}}$ &
$\substack{\frac{(p^{m_2}-1)}{2}(p^{m-m_2}-\frac{(-1)^{l_4}p^{\frac{m+m_2}{2}}}{p^{m_2}}) \\ \cdot(p^{m_2-m_1}-\frac{(-1)^{l_7}\sqrt[]{p^{m_2}}(p^{m_1}-1)}{p^{m_1}})}$\\
$p^{m-m_1}+\frac{(-1)^{l_2}\sqrt[]{p^m}(p^{m_1}-1)}{p^{m_1}}-\frac{p^m}{p^{m_1+m_2}}$ &
$(p^{m-m_2}-1)(p^{m_2}-1)$\\
$p^{m-m_1}+\frac{(-1)^{l_2}\sqrt[]{p^m}(p^{m_1}-1)}{p^{m_1}}-\frac{p^m-(-1)^{l_4}p^{\frac{m+m_2}{2}}}{p^{m_1+m_2}}$ &
$\substack{\frac{(p^{m_2}-1)}{2}(p^{m-m_2}+\frac{(-1)^{l_4}p^{\frac{m+m_2}{2}}}{p^{m_2}}) \\ \cdot(p^{m_2}-p^{m_2-m_1}-\frac{(-1)^{l_7}\sqrt[]{p^{m_2}}(p^{m_1}-1)}{p^{m_1}})}$\\
$p^{m-m_1}+\frac{(-1)^{l_2}\sqrt[]{p^m}(p^{m_1}-1)}{p^{m_1}}-\frac{p^m+(-1)^{l_4}p^{\frac{m+m_2}{2}}}{p^{m_1+m_2}}$ &
$\substack{\frac{(p^{m_2}-1)}{2}(p^{m-m_2}-\frac{(-1)^{l_4}p^{\frac{m+m_2}{2}}}{p^{m_2}}) \\ \cdot(p^{m_2}-p^{m_2-m_1}+\frac{(-1)^{l_7}\sqrt[]{p^{m_2}}(p^{m_1}-1)}{p^{m_1}})}$\\

\bottomrule
\end{tabular}
\end{center}
\end{table}

\begin{proof}
Similarly to the proof of Theorem \ref{th-4.1}, the desired conclusions follow from Lemmas \ref{lem-omega}, \ref{lem-N}, \ref{lem-f} and Lemmas \ref{LH}, \ref{lem-2N}, \ref{lem-weil}.
\end{proof}

We use Magma to give some examples in the following which are consistent
with the results in Theorem \ref{th-4.2}.

\begin{example}
Let $(p,m,m_1,m_2)=(3,5,1,1)$. Then the code $\overline{\cC_{D}}$ in Theorem \ref{th-4.2} has parameters $[81,6,48]$ and weight enumerator $1+90z^{48}+144z^{51}+240z^{54}+180z^{57}+72z^{60}+2z^{81}$. Since $\overline{\cC_{D}}$ is 3-divisible, we deduce $\overline{\cC_{D}}$ is self-orthogonal by Lemma \ref{LH}.
\end{example}

\begin{example}\label{exa-optimal}
Let $(p,m,m_1,m_2)=(3,4,1,1)$. Then the code $\overline{\cC_{D}}$ in Theorem \ref{th-4.2} has parameters $[21,5,12]$ and weight enumerator $1+100z^{12}+120z^{15}+20z^{18}+2z^{21}$. This code is optimal with respect to the Griesmer bound.  Since $\overline{\cC_{D}}$ is 3-divisible, we deduce $\overline{\cC_{D}}$ is self-orthogonal by Lemma \ref{LH}.
\end{example}

\begin{example}
Let $(p,m,m_1,m_2)=(3,6,1,2)$. Then the code $\overline{\cC_{D}}$ in Theorem \ref{th-4.2} has parameters $[261,4,216]$ and weight enumerator $1+80z^{216}+1800z^{228}+2304z^{231}+640z^{234}+1440z^{237}+288z^{240}+8z^{261}$.  Since $\overline{\cC_{D}}$ is 3-divisible, we deduce $\overline{\cC_{D}}$ is self-orthogonal by Lemma \ref{LH}.
\end{example}

\section{The locality of $\overline{\cC_D}$ and the minimum distance of $\overline{\cC_D}^\perp$}
In this section, we determine the locality of $\overline{\cC_D}$ and the minimum distance of $\overline{\cC_D}^\perp$.

\begin{theorem}\label{th-locality}
 Let $p$ be an odd prime. Let $m, m_1, m_2$ be three positive integers such that $m_1 \mid m$ and $m_2 \mid m$. Then  $\overline{\cC_D}$ is a locally recoverable code with locality $2$ and $\overline{\cC_D}^\perp$ has minimum distance $3$.
\end{theorem}

\begin{proof}
  Let $\gf_{p^m}^*=\langle\alpha\rangle$. Then $\left\{\alpha^0, \alpha^1, \cdots, \alpha^{\frac{m}{m_2}-1}\right\}$ is a $\gf_{p^{m_2}}$-basis of $\gf_{p^m}$. Let $d_1, d_2, \cdots, d_{n-1},d_n$ be all the elements in $D$. For convenience, let  $d_n = 0$ due to $0\in D$.
By definition, the generator matrix $G$ of $\overline{\cC_D}$ is given by
\begin{eqnarray*}
G:=\left[
\begin{array}{cccc}
1&1&\cdots&1 \\
 \tr_{p^m/p^{m_2}}(\alpha^{0}d_{1})& \tr_{p^m/p^{m_2}}(\alpha^{0}d_{2})& \cdots &\tr_{p^m/p^{m_2}}(\alpha^{0}d_{n}) \\
\tr_{p^m/p^{m_2}}(\alpha^{1}d_{1})& \tr_{p^m/p^{m_2}}(\alpha^{1}d_{2})& \cdots &\tr_{p^m/p^{m_2}}(\alpha^{1}d_{n}) \\
\vdots &\vdots &\ddots &\vdots \\
\tr_{p^m/p^{m_2}}(\alpha^{\frac{m}{m_2}-1}d_{1})& \tr_{p^m/p^{m_2}}(\alpha^{\frac{m}{m_2}-1}d_{2})& \cdots &\tr_{p^m/p^{m_2}}(\alpha^{\frac{m}{m_2}-1}d_{n})
\end{array}\right].
\end{eqnarray*}

In the following, we first determine the locality of $\overline{\cC_D}$.
For convenience, we assume that $\bg_i$ denotes the $i$-th column of $G$, i.e., $$\bg_i=(1, \tr_{p^m/p^{m_2}}(\alpha^0d_i), \tr_{p^m/p^{m_2}}(\alpha^1 d_i), \cdots, \tr_{p^m/p^{m_2}}(\alpha^{\frac{m}{m_2}-1}d_i))^T,$$ where $1 \leq i \leq n$.
Then
$$\bg_n=(1,0,0,\cdots,0)^T.$$
For any $d_i$ with $1\leq i \leq n-1$ and $d_n=0$, if we choose any $u\in \gf_{p^{m_2}}\backslash \{0,1\}$, $v=1-u$ and $d_j=u^{-1}d_i$, then we can verify that
the following holds:
 \begin{eqnarray*}
\left\{\begin{array}{lll}
1&=&u+v,\\
 \tr_{p^m/p^{m_2}}(\alpha^0d_i) &=&u\tr_{p^m/p^{m_2}}(\alpha^{0}d_j)+v\tr_{p^m/p^{m_2}}(\alpha^{0}d_n), \\
  \tr_{p^m/p^{m_2}}(\alpha^1d_i)&=&u\tr_{p^m/p^{m_2}}(\alpha^{1}d_j)+v\tr_{p^m/p^{m_2}}(\alpha^{1}d_n), \\
 & \vdots &\\
  \tr_{p^m/p^{m_2}}(\alpha^{\frac{m}{m_2}-1}d_{i})&=&u\tr_{p^m/p^{m_2}}(\alpha^{\frac{m}{m_2}-1}d_j)+v\tr_{p^m/p^{m_2}}(\alpha^{\frac{m}{m_2}-1}d_n).
\end{array}\right.
\end{eqnarray*}
In other words, for any column $\bg_i$ with $1\leq i \leq n-1$, there exist two other columns $\bg_j$ and $\bg_n$ such that $\bg_i$ is a linear combination of them.
Besides, $\bg_n$ is also a linear combination of such $\bg_i$ and $\bg_j$.
Hence $\overline{\cC_D}$ has locality $2$.

Now we determine the minimum distance $d^\perp$ of $\overline{\cC_D}^\perp$. We claim that any two columns of $G$ are $\gf_{p^{m_2}}$-linearly independent.
In fact, if there exist two columns $\bg_i$ and $\bg_j$ for $1 \leq i \neq j \leq n$ which are $\gf_{p^{m_2}}$-linearly dependent, then their exist $u,v\in \gf_{p^{m_2}}^*$
such that
 \begin{eqnarray*}
\left\{\begin{array}{lll}
0&=&u+v,\\
 0 &=&u\tr_{p^m/p^{m_2}}(\alpha^{0}d_i)+v\tr_{p^m/p^{m_2}}(\alpha^{0}d_j), \\
 0 &=&u\tr_{p^m/p^{m_2}}(\alpha^{1}d_i)+v\tr_{p^m/p^{m_2}}(\alpha^{1}d_j), \\
 & \vdots &\\
 0 &=&u\tr_{p^m/p^{m_2}}(\alpha^{\frac{m}{m_2}-1}d_i)+v\tr_{p^m/p^{m_2}}(\alpha^{\frac{m}{m_2}-1}d_j).
\end{array}\right.
\end{eqnarray*}
Then we have $\tr_{p^m/p^{m_2}}(\alpha^{t}(d_i-d_j))=0$ for all $0\leq t \leq \frac{m}{m_2}-1$.
For any $x=\sum_{t=0}^{\frac{m}{m_2}-1}x_i\alpha^i\in \gf_{p^m}$ and each $x_i\in \gf_{p^{m_2}}$, we then derive
\begin{eqnarray*}& &\tr_{p^m/p^{m_2}}(x(d_i-d_j))\\
&=&\tr_{p^m/p^{m_2}}\left(\sum_{t=0}^{\frac{m}{m_2}-1}x_t\alpha^t(d_i-d_j)\right)\\
&=&\sum_{t=0}^{\frac{m}{m_2}-1}x_t\tr_{p^m/p^{m_2}}(\alpha^t(d_i-d_j))=0.
\end{eqnarray*}
This contradicts with $\left \lvert \ker(\tr_{p^m/p^{m_2}}) \right\rvert =p^{m-m_2}<p^m$. Hence any two columns of $G$ are $\gf_{p^{m_2}}$-linearly independent.
Then we have $d^\perp\geq 3$. Besides, since there exist three columns of $G$ which are $\gf_{p^{m_2}}$-linearly dependent,
we have $d^\perp=3$.
\end{proof}

By Theorems \ref{th-4.1}, \ref{th-4.2} and Theorem \ref{th-locality}, we have the following corollary.

\begin{corollary}
The dual of $\overline{\cC_D}$ in Theorem \ref{th-4.1} or \ref{th-4.2} is at least almost optimal with respect to the sphere-packing bound.
\end{corollary}

\begin{proof}
By the sphere-packing bound, it is easy to prove that their exists no linear code over $\gf_{p^{m_2}}$ of the same length and dimension as those of  $\overline{\cC_D}^\perp$
and minimum distance 5. Then the desired conclusion follows.
\end{proof}

By the Code Tables at http://www.codetables.de/, we find that the dual of $\overline{\cC_D}$ in Theorem \ref{th-4.1} or \ref{th-4.2} is optimal or has best known parameters in some cases and is almost optimal in some other cases.

\begin{table}[h]
\begin{center}
\caption{The optimality of $\overline{\cC_D}^{\perp}$.}\label{tab-optimalcode}
\begin{tabular}{llll}
\toprule
Conditions & Code & Parameters & Optimality \\
\midrule
$p=3,m=4,m_1=m_2=1$ & $\overline{\cC_D}^{\perp}$ & $[21, 16, 3]$ & Optimal\\
$p=3,m=5,m_1=m_2=1$ & $\overline{\cC_D}^{\perp}$ & $[81,75,3]$ & Optimal\\
$p=3,m=4,m_1=1,m_2=2$ & $\overline{\cC_D}^{\perp}$ & $[21,18,3]$ & Optimal\\
$p=5,m=4,m_1=m_2=1$ & $\overline{\cC_D}^{\perp}$ & $[105, 100, 3]$ & Best known parameters\\
$p=3,m=3,m_1=m_2=1$ & $\overline{\cC_D}^{\perp}$ & $[9, 5, 3]$ & Almost optimal\\
$p=5,m=3,m_1=m_2=1$ & $\overline{\cC_D}^{\perp}$ & $[25, 21, 3]$ & Almost optimal\\
$p=7,m=3,m_1=m_2=1$ & $\overline{\cC_D}^{\perp}$ & $[49, 45, 3]$ & Almost optimal\\
$p=3,m=6,m_1=3,m_2=1$ & $\overline{\cC_D}^{\perp}$ & $[53,46,3]$ & Almost optimal\\
$p=3,m=6,m_1=2,m_2=1$ & $\overline{\cC_D}^{\perp}$ & $[81,74,3]$ & Almost optimal\\
$p=3,m=6,m_1=m_2=2$ & $\overline{\cC_D}^{\perp}$ & $[81, 77, 3]$ & Almost optimal\\
\bottomrule
\end{tabular}
\end{center}
\end{table}

A linear code is said to be projective if its dual has minimum distance at least 3.
Theorem \ref{th-locality} implies the following corollary.
\begin{corollary}\label{cor}
The linear code $\overline{\cC_D}$ in Theorem \ref{th-4.1} or \ref{th-4.2} is projective.
\end{corollary}

In the following, we present some optimal or almost locally recoverable codes from $\overline{\cC_D}$.

\begin{example}\label{optimal-1}
Let $m=3,m_1=m_2=1$ and $p=3$. Then $\overline{\cC_D}$ has parameters $[9,4,4]$ by Theorem \ref{th-4.1} and locality $r=2$ by Theorem \ref{th-locality}.
For $n=9,d=4,r=2,q=3$, by the Griesmer bound we have
\begin{eqnarray*}
k_{\text{opt}}^{(q)}(n-(r+1),d)=k_{\text{opt}}^{(q)}(6,4)=2.
\end{eqnarray*}
By the right hand of Cadambe-Mazumdar bound in Lemma \ref{lem-CMbound}, we have
\begin{eqnarray*}\label{eqn-CMbound}
\mathop{\min}_{t \in \mathbb{Z}^{+}}\left [rt+k_{\text{opt}}^{(q)}(n-t(r+1),d)\right]
\leq \mathop{\min}_{t=1}\left [rt+k_{\text{opt}}^{(q)}(n-t(r+1),d)\right]=4,
\end{eqnarray*}
which implies that $\overline{\cC_D}$ is a $[9,4,4; 2]_3$ \textbf{$k$-optimal locally recoverable code} according to the Cadambe-Mazumdar bound.
By the right hand of the Singleton-like bound in Lemma \ref{like}, we have
\begin{eqnarray*}
n-k-\left\lceil\frac{k}{r} \right\rceil+2=5=d+1,
\end{eqnarray*}
which implies that  $\overline{\cC_D}$ is a $[9,4,4; 2]_3$ \textbf{almost $d$-optimal locally recoverable code} according to the Singleton-like bound.
\end{example}

\begin{example}\label{optimal-2}
Let $m=4,m_1=m_2=1$ and $p=3$. Then $\overline{\cC_D}$ has parameters $[21,5,12]$ by Theorem \ref{th-4.1} and locality $2$ by Theorem \ref{th-locality}.
For $n=21,d=12,r=2,q=3$, by the Griesmer bound, we have
\begin{eqnarray*}
k_{\text{opt}}^{(q)}(n-(r+1),d)=k_{\text{opt}}^{(q)}(18,12)=3.
\end{eqnarray*}
By the right hand of Cadambe-Mazumdar bound in Lemma \ref{lem-CMbound}, we have
\begin{eqnarray*}\label{eqn-CMbound}
\mathop{\min}_{t \in \mathbb{Z}^{+}}\left [rt+k_{\text{opt}}^{(q)}(n-t(r+1),d)\right]
\leq \mathop{\min}_{t=1}\left [rt+k_{\text{opt}}^{(q)}(n-t(r+1),d)\right]=5,
\end{eqnarray*}
which implies that $\overline{\cC_D}$ is a $[21,5,12; 2]_3$ \textbf{$k$-optimal locally recoverable code} according to the Cadambe-Mazumdar bound.
\end{example}

\begin{example}\label{optimal-3}
Let $m=4,m_1=1,m_2=2$ and $p=3$. Then $\overline{\cC_D}$ has parameters $[21,3,16]$ by Theorem \ref{th-4.2} and locality $2$ by Theorem \ref{th-locality}.
For $n=21,d=16,r=2$ and $q=9$, by the Griesmer bound we have
\begin{eqnarray*}
k_{\text{opt}}^{(q)}(n-(r+1),d)=k_{\text{opt}}^{(q)}(18,16)=2.
\end{eqnarray*}
By the right hand of Cadambe-Mazumdar bound in Lemma \ref{lem-CMbound}, we have
\begin{eqnarray*}\label{eqn-CMbound}
\mathop{\min}_{t \in \mathbb{Z}^{+}}\left [rt+k_{\text{opt}}^{(q)}(n-t(r+1),d)\right]
= \mathop{\min}_{t=1}\left [rt+k_{\text{opt}}^{(q)}(n-t(r+1),d)\right]=4
\end{eqnarray*}
as $n-t(r+1)\geq d$ iff $t=1$.
Then $\overline{\cC_D}$ is a $[21,3,16; 2]_9$ \textbf{almost $k$-optimal locally recoverable code} according to the Cadambe-Mazumdar bound.
\end{example}

\section{Concluding remarks}\label{sec6}
In this paper, inspired  by the work in \cite{D4}, we presented a new construction of linear codes  $\overline{\cC_D}$ based on the function $f(x)=x^2$.
Using quadratic Gaussian sums, we derived the parameters and weight distributions of $\overline{\cC_D}$ in several cases.
In this cases, $\overline{\cC_D}$ has only a few weights and was proved to be both $p$-divisible and self-orthogonal.
Besides, it was also proved that $\overline{\cC_D}$ is projective and has locality only 2.
In particular, we derived some optimal or almost optimal linear codes in Example \ref{exa-optimal} and Table \ref{tab-optimalcode}.
We also obtained some optimal or almost optimal locally recoverable codes in Examples \ref{optimal-1}, \ref{optimal-2} and \ref{optimal-3}.
These results indicate that our codes not only have good parameters but also have nice applications in quantum codes, lattices and distributed storage.

Compared with the construction by Ding and Ding in \cite[Theorems 1 and 2]{D4}, the construction given in this paper has the following advantages.
\begin{enumerate}
\item The construction given in this paper can produce linear codes with higher code rate. For instance, if $(p,m,m_1,m_2) = (3,6,2,1)$, then our code $\overline{\cC_D}$ is a $[81,7,48]$ ternary code, while the code in \cite[Example 1]{D4} is a $[80, 5, 48]$ ternary code; if $(p,m,m_1,m_2)=(5,4,1,1)$, then our code is a $[105, 5, 80]$ code over $\gf_5$, while the code in \cite[Example 2]{D4} is a $[104, 4, 80]$ code over $\gf_5$. Hence, our codes improve the parameters of those in \cite{D4}.
\item Our codes are projective by Corollary \ref{cor}, while the codes in \cite[Theorems 1 and 2]{D4} are not projective.
\end{enumerate}

\section*{}
\textbf{Funding}
Z. Heng's research was supported in part by the National Natural Science Foundation of China under Grant 12271059, in part by the open research
fund of National Mobile Communications Research Laboratory£¬Southeast University (No. 2024D10), and in part by the Fundamental Research Funds for the Central Universities, CHD, under Grant 300102122202.

\end{document}